                            \documentclass[11pt]{article}
                            \usepackage{amssymb,amsmath}
                            \usepackage{xspace,color,comment}
                            \usepackage{graphicx}

                            \usepackage{xspace}
                            \usepackage{mathptmx}
                            \usepackage{epsfig}
                            \usepackage{epstopdf}
                            \usepackage{graphicx}
                            \usepackage{amsfonts}
                            \usepackage{amsmath}
                            \usepackage{amsthm}
                            \usepackage{proof}
                            \usepackage{latexsym}
                            \usepackage{amssymb}
                            \usepackage{color}
                            \usepackage{comment}
                            \usepackage[margin=1in]{geometry}
                            \usepackage{extras,tikz}
                            \usepackage{psfrag}
                            \usepackage{comment}
                            \usepackage[noend]{algorithmic}
                            \usepackage{algorithm}
                            \usepackage{tweaklist}

                            \usepackage{url}

                            \newtheorem{theorem}{Theorem}[section]
                            \newtheorem{lemma}[theorem]{Lemma}

                            \newtheorem{remark}[theorem]{Remark}

                            \newtheorem{definition}[theorem]{Definition}
                            
                            \DeclareMathAlphabet{\mathcal}{OMS}{cmsy}{m}{n}
                            \usepackage{tweaklist}
                            \renewcommand{\enumhook}{\setlength{\itemsep}{-0.5ex}}
                            \renewcommand{\itemhook}{\setlength{\itemsep}{-0.5ex}}

                            \newtheorem*{rep@theorem}{\rep@title} \newcommand{\newreptheorem}[2]{%
                            \newenvironment{rep#1}[1]{%
                            \def\rep@title{\bf #2 \ref{##1} }%
                            \begin{rep@theorem} }%
                            {\end{rep@theorem} } }
                            \makeatother
                            \newreptheorem{theorem}{Theorem}
                            \newreptheorem{lemma}{Lemma}

                            \newcommand*\samethanks[1][\value{footnote}]{\footnotemark[#1]}

\newcommand{\moves}{ \mathcal{M} } 
\newcommand{\tokens}{ \mathcal{T} } 
\newcommand{\fp}{ F } 
\newcommand{\FP}{ \mathbb{\textbf{F}} } 
\newcommand{\rank}{ {\mathcal{R}} } 
\newcommand{\states}{V_{H}^{|\tokens|}}
\newcommand{\state}{V^{k+1}}

\newcommand{\ins}{\text{in}}
\newcommand{\out}{\text{out}}

\newcommand{\s}{\mathcal{S}}

\newcommand{\costtoken}{\hat{C}} 

\newcommand{\bp}{ B } 
\newcommand{\BP}{ \mathbb{\textbf{B}} } 

\newcommand{\opt}{\text{OPT} } 

\newcommand{\f}{\textbf{f} } 
\newcommand{\bt}{\textbf{b} } 

\newcommand{\gt}{\textsc{Grid Tiling}\xspace}

\newcommand{\scss}{\textsc{Strongly Connected Steiner Subgraph} \xspace}

\newcommand{\blah}{$2$-SCSS-$(k,1)$\xspace}
\newcommand{\blaha}{$2$-SCSS-$(2k-1,1)$\xspace}
\newcommand{\blahgeneral}{$2$-SCSS-$(k_1, k_2)$\xspace}

\newcommand{\sndp}{\textsc{Dir-Cap-SNDP}\xspace}

\begin{document}


                            \title{A Tight Algorithm for Strongly Connected Steiner Subgraph On Two Terminals With Demands\thanks{An extended abstract~\cite{ipec} appeared in IPEC '14}}


                            %
                            %
                            \author{Rajesh Chitnis\thanks{The Weizmann Institute of Science, Rehovot, Israel. Supported by a postdoctoral fellowship from I-CORE ALGO. Work done in part when at the University of Maryland, College Park. Email: \texttt{rajesh.chitnis@weizmann.ac.il}} \and Hossein Esfandiari\thanks{Department of Computer Science, University of Maryland at College Park, USA. Supported in part by NSF CAREER award 1053605, NSF grant CCF-1161626, ONR YIP award N000141110662, DARPA/AFOSR grant FA9550-12-1-0423. Email: \texttt{\{hossein, hajiagha, saeedrez\}@cs.umd.edu}} \and MohammadTaghi Hajiaghayi\samethanks[3] \and Rohit Khandekar\thanks{KCG Holdings Inc., USA. Email: \texttt{rkhandekar@gmail.com}} \and Guy Kortsarz\thanks{ Department of Computer Science, Rutgers University-Camden, USA. Supported by NSF grant 1218620. Email: \texttt{guyk@camden.rutgers.edu}} \and Saeed Seddighin\samethanks[3]}%
                            %



                            %
                            %

                            %
                            %


\maketitle
\thispagestyle{empty}

\begin{abstract}
Given an edge-weighted directed graph $G=(V,E)$ on $n$ vertices and a set $T=\{t_1, t_2, \ldots, t_p\}$ of $p$ terminals, the objective of the \scss ($p$-SCSS)
problem is to find an edge set $H\subseteq E$ of minimum weight such that $G[H]$ contains an $t_{i}\rightarrow t_j$
path for each $1\leq i\neq j\leq p$. The $p$-SCSS problem is NP-hard, but Feldman and Ruhl [FOCS '99; SICOMP '06] gave a novel $n^{O(p)}$ time
algorithm.

In this paper, we investigate the computational complexity of a variant of $2$-SCSS where we have demands for the number of paths between each terminal pair. Formally, the \blahgeneral problem is defined as follows: given an edge-weighted directed graph $G=(V,E)$ with weight function $\omega: E\rightarrow \mathbb{R}^{\geq 0}$, two terminal vertices $s, t$, and integers $k_1, k_2$ ; the objective is to find a set of $k_1$ paths $F_1, F_2, \ldots, F_{k_1}$ from $s\leadsto t$ and $k_2$ paths $B_1, B_2, \ldots, B_{k_2}$ from $t\leadsto s$ such that $\sum_{e\in E} \omega(e)\cdot \phi(e)$ is minimized, where $\phi(e)= \max \Big\{|\{i\in [k_1] : e\in F_i\}|\ ,\ |\{j\in [k_2] : e\in B_j\}|\Big\}$. For each $k\geq 1$, we show the following:
\begin{itemize}
\item The \blah problem can be solved in time $n^{O(k)}$.
\item A matching lower bound for our algorithm: the \blah problem does not have an $f(k)\cdot n^{o(k)}$ time algorithm for any computable function $f$, unless the Exponential Time Hypothesis (ETH) fails.
\end{itemize}
Our algorithm for \blah relies on a structural result regarding an optimal solution followed by using the idea of a ``token game" similar to that of Feldman and Ruhl. We show with an example that the structural result does not hold for the \blahgeneral problem if $\min\{k_1, k_2\}\geq 2$. Therefore \blah is the most general problem one can attempt to solve with our techniques. To obtain the lower bound matching the algorithm, we reduce from a special variant of the \textsc{Grid Tiling} problem introduced by Marx [FOCS '07; ICALP '12].


\end{abstract}

\newpage

\section{Introduction}
\label{sec:intro}

The \textsc{Steiner Tree} (ST) problem is one of the earliest and
most fundamental problems in combinatorial optimization: given an
undirected edge-weighted graph $G=(V,E)$ with edge weights $\omega: E\rightarrow \mathbb{R}^{\geq 0}$ and a set $T\subseteq V$ of terminals, the
objective is to find a tree $S$ of minimum weight $\omega(S):=\sum_{e\in S} \omega(e)$  which spans all the terminals. The \textsc{Steiner Tree} problem is believed to have been
first formally defined by Gauss in a letter in 1836.
In the directed version, called the \textsc{Directed Steiner Tree} (DST) problem, we are also given a root vertex $r$ and the objective is to find a minimum
size arborescence in the directed graph which connects the root $r$ to each terminal from $T$. An easy reduction from \textsc{Set Cover} shows that the DST
problem is also NP-complete.

Steiner-type of problems arise in the design of networks. Since many networks are symmetric, the directed versions of Steiner
type of problems were mostly of theoretical interest. However in recent years, it has been
observed~\cite{ramanathan1996multicast} that the connection cost in various networks such as satellite
or radio networks are not symmetric. Therefore, directed graphs are the most suitable model for such networks. In addition,
Ramanathan~\cite{ramanathan1996multicast} also used the DST problem to find low-cost multicast trees, which have applications
in point-to-multipoint communication in high bandwidth networks.
If we require two-way connectivity, then we obtain a generalization of the DST problem known as the \scss (SCSS) problem. In the $p$-SCSS problem, given a directed graph
$G=(V,E)$ and a set $T=\{t_{1}, t_{2}, \ldots, t_{p}\}$ of $p$ terminals the objective is to find a set $H\subseteq E$ of minimum size such that $G[H]$ contains an $t_{i}\rightarrow t_{j}$ path for each $1\leq i\neq j\leq p$. The SCSS problem is also NP-hard. The best known approximation ratio in polynomial time for SCSS is $|T|^{\epsilon}$ for any $\epsilon>0$ due to Charikar et al.~\cite{DBLP:journals/jal/CharikarCCDGGL99}. A result
of Halperin and Krauthgamer~\cite{DBLP:conf/stoc/HalperinK03} implies SCSS has no $\Omega(\log^{2-\epsilon} n)$-approximation
for any $\epsilon>0$, unless NP has quasi-polynomial Las Vegas algorithms. Regarding exact algorithms, Feldman and Ruhl~\cite{feldman-ruhl} gave a novel $n^{O(p)}$ time algorithm for $p$-SCSS. Chitnis et al.~\cite{soda14} showed that this algorithm is almost tight by the following result: for any computable function $f$, the $p$-SCSS problem has no $f(p)\cdot n^{o(p/\log p)}$ time algorithm unless the Exponential Time Hypothesis (ETH) fails. Chitnis et al.~\cite{soda14} showed that on certain special graph classes such as planar graphs (and more generally $H$-minor-free graphs) one can obtain faster algorithms than that of Feldman and Ruhl: more specifically, if the underlying undirected graph is planar, then $p$-SCSS can be solved in time $2^{O(p\log p)}\cdot n^{O(\sqrt{p})}$. In addition, Chitnis et al.~\cite{soda14} also showed that this algorithm is optimal: for any computable function $f$, the existence of a $f(p)\cdot n^{o(\sqrt{p})}$ time algorithm for $p$-SCSS on planar graphs implies ETH fails.

\vspace{3mm}

\noindent \textbf{The \blahgeneral Problem:} We now define the following generalization of the $2$-SCSS problem:
\begin{center}
\noindent\framebox{\begin{minipage}{6.00in}
\textbf{\blahgeneral}\\
\emph{\underline{Input} }: An edge-weighted digraph $G=(V,E)$ with weight function $\omega: E\rightarrow \mathbb{R}^{\geq 0}$, two terminal vertices $s, t$, and integers $k_1, k_2$ \\
\emph{\underline{Question}}: Find a set of $k_1$ paths $F_1, F_2, \ldots, F_{k_1}$ from $s\leadsto t$ and $k_2$ paths $B_1, B_2, \ldots, B_{k_2}$ from $t\leadsto s$ such that $\sum_{e\in E} \omega(e)\cdot \phi(e)$ is minimized where $\phi(e)= \max \Big\{|\{i\in [k_1] : e\in F_i\}|\ ,\ |\{j\in [k_2] : e\in B_j\}|\Big\}$.

\end{minipage}}
\end{center}
Observe that $2$-SCSS-$(1, 1)$ is the same as the $2$-SCSS problem. The definition of the \blahgeneral problem allows us to potentially choose the same edge multiple times, but we have to pay for each time we use it in a path between a given terminal pair. This can be thought of as \textbf{``buying disjointness"} by adding parallel edges. In large real-world networks, it might be more feasible to modify the network by adding some parallel edges to create disjoint paths than finding disjoint paths in the existing network. Teixeira et al.~\cite{DBLP:conf/sigmetrics/TeixeiraMSV03,DBLP:conf/imc/TeixeiraMSV03} model path diversity in Internet Service Provider (ISP) networks and the Sprint network by disjoint paths between two hosts. There have been several patents~\cite{guo2011hybrid,ramachandran2010wireless} attempting to design multiple paths between the components of Google Data Centers.

The \blahgeneral problem is a special case of the \textsc{Directed Survivable Network Design} (\sndp) problem~\cite{DBLP:conf/soda/GoemansGPSTW94} in which we are given an directed multigraph with weights and capacities on the edges, and the question is to find a minimum weight subset of edges that satisfies all pairwise minimum-cut requirements. In the \blahgeneral problem, we \textbf{do not} require disjoint paths.
As observed in Chakrabarty et al.~\cite{DBLP:conf/ipco/ChakrabartyCKK11} and Goemans et al.~\cite{DBLP:conf/soda/GoemansGPSTW94}, the \sndp problem becomes much easier to approximate if we allow taking multiple copies of each edge.

\subsection{Our Results and Techniques:}
\label{subsec:our-results}

In this paper, we consider the \blah problem parameterized by $k$. Note that the sum of demands is $O(k)$. To the best of our knowledge, we are unaware of any non-trivial exact algorithms for a version of the SCSS problem with demands between the terminal pairs.
Our main algorithmic result is the following:

\begin{theorem}
\label{thm:main-algorithmic} The \blah problem can be solved in $n^{O(k)}$ time, where $n$ is the number of vertices in the input graph.
\end{theorem}

Our algorithm proceeds as follows: In Section~\ref{subsec:structure} we first show that there is an optimal solution for the \blah problem which satisfies a structural property which we call as \textbf{reverse-compatibility}. Then in Section~\ref{tokengame} we introduce a ``Token Game" (similar to that of Feldman and Ruhl~\cite{feldman-ruhl}), and show that the \textsc{Solving-Token-Game} problem can be solved in $n^{O(k)}$ time. Finally in Section~\ref{reduction}, using the existence of an optimal solution satisfying reverse-compatibility, we give a reduction from the \blah problem to the \textsc{Solving-Token -Game} problem which gives an $n^{O(k)}$ time algorithm for the \blah problem. This algorithm also generalizes the result of Feldman and Ruhl~\cite{feldman-ruhl} for $2$-SCSS, since $2$-SCSS is equivalent to $2$-SCSS-$(1, 1)$. In Section~\ref{app:no-structure}, we show with an example (see Figure~\ref{fig:counter}) that the structural result does not hold for the \blahgeneral problem if $\min\{k_1, k_2\}\geq 2$. Therefore, \blah is the most general problem that one can attempt to solve with our techniques.

Theorem~\ref{thm:main-algorithmic} does not rule out the possibility that the \blah problem is actually solvable in polynomial time. Our main hardness result rules out this possibility by showing that our algorithm is \emph{tight} in the sense that the exponent of $O(k)$ is best possible.
\begin{theorem}
\label{thm:main-hardness}
The \blah problem is W[1]-hard parameterized by $k$. Moreover, under the Exponential Time Hypothesis (ETH) of Impagliazzo and Paturi~\cite{eth}, the \blah problem cannot be solved in $f(k)\cdot n^{o(k)}$ time for any computable function $f$ where $n$ is the number of vertices in the graph.
\end{theorem}

To prove Theorem~\ref{thm:main-hardness}, we reduce from the \gt problem formulated in the pioneering work of Marx ~\cite{daniel-grid-tiling}:
\begin{center}
\noindent\framebox{\begin{minipage}{4.50in}
\textbf{\textsc{$k\times k$ Grid Tiling}}\\
\emph{Input }: Integers $k, n$, and $k^2$ non-empty sets $S_{i,j}\subseteq [n]\times [n]$ where $i, j\in [k]$\\
\emph{Question}: For each $1\leq i, j\leq k$ does there exist a value $s_{i,j}\in S_{i,j}$ such that
\begin{itemize}
\item If $s_{i,j}=(x,y)$ and $s_{i,j+1}=(x',y')$ then $x=x'$.
\item If $s_{i,j}=(x,y)$ and $s_{i+1,j}=(x',y')$ then $y=y'$.
\end{itemize}
\end{minipage}}
\end{center}
The \gt problem has turned to be a convenient starting point for parameterized reductions for planar problems, and has been used recently in various W[1]-hardness proofs on planar graphs~\cite{soda14,DBLP:conf/icalp/Marx12,michal-stacs,daniel-voronoi}. Under the ETH, Chen et al.~\cite{chen-hardness} showed that $k$-\textsc{Clique}\footnote{The $k$-\textsc{Clique} problem asks whether there is a clique of size $\geq k$?} does not admit an algorithm running in time $f(k)\cdot n^{o(k)}$ for any computable function $f$. Marx~\cite{daniel-grid-tiling} gave a reduction from $k$-\textsc{Clique} to $k\times k$ \gt. In Section~\ref{sec:hardness}, we give a reduction from $k\times k$ \gt to \blaha. Since the parameter blowup is linear, the $f(k)\cdot n^{o(k)}$ lower bound for \gt from~\cite{daniel-grid-tiling} transfers to \blah.

Before proceeding further, we show that the edge-weighted and the vertex-weighted variants of \blahgeneral are computationally equivalent. First we define the vertex-weighted variant of \blahgeneral.

\begin{center}
\noindent\framebox{\begin{minipage}{6.00in}
\textbf{Vertex-weighted \blahgeneral}\\
\emph{\underline{Input} }: A vertex-weighted digraph $G=(V,E)$ with weight function $\omega': V\rightarrow \mathbb{R}^{\geq 0}$, two terminal vertices $s, t$, and integers $k_1, k_2$ \\
\emph{\underline{Question}}: Find a set of $k_1$ paths $F_1, F_2, \ldots, F_{k_1}$ from $s\leadsto t$ and $k_2$ paths $B_1, B_2, \ldots, B_{k_2}$ from $t\leadsto s$ such that $\sum_{v\in V\setminus \{s,t\}} \omega'(v)\cdot \phi'(v)$ is minimized where $\phi'(v)= \max \Big\{|\{i\in [k_1] : v\in F_i\}|\ ,\ |\{j\in [k_2] : v\in B_j\}|\Big\}$.

\end{minipage}}
\end{center}

\begin{lemma} \label{limu}
The edge-weighted \blahgeneral and the vertex-weighted \blahgeneral are equivalent.
\end{lemma}
\begin{proof}
First, we show that the edge-weighted version can be solved using the vertex-weighted version. Let $G=(V,E)$ be an edge-weighted graph with weight function $\omega$. We create a vertex-weighted graph $G'=(V',E')$ with weight function $\omega'$ as follows: subdivide each edge $(u,v)$ of $G$ by adding a new vertex $\beta_{u,v}$ to get a path $u\rightarrow \beta_{u,v}\rightarrow v$ of length two. Let $V'=V\cup \{\beta_{u,v} : (u,v)\in E\}$. Set $\omega'(v)=0$ for each $v\in V$ and $\omega'(\beta_{u,v})=\omega(u,v)$ for each edge $(u,v)\in E$. Consider any solution $H\subseteq E$ of edge-weighted \blahgeneral. Consider the solution $H'$ obtained by including (the subdivision) of each edge in $H$. The weights of all vertices from $V$ is zero in $G'$. Also, for any edge $e=(u,v)$ it is easy to see that $\phi(e)=\phi'(\beta_{u,v})$, and hence both solutions $H$ and $H'$ have same cost.

Next, we show that the vertex-weighted version can be solved using the edge-weighted version. Let $G'=(V',E')$ be a vertex-weighted graph with weight function $\omega'$. We create an edge-weighted graph $G=(V,E)$ with weight function $\omega$ as follows: Replace each vertex $v\in V\setminus \{s,t\}$  with a pair of vertices $(v_{\ins},v_{\out})$. Let $s_{\ins}=s=s_{\out}$ and $t_{\ins}=t=t_{\out}$. Make all in-neighbors, out-neighbors of $v$ in $G$ incident to $v_{\ins}, v_{\out}$ respectively and add an edge $(v_{\ins}, v_{\out})$. Set $\omega(v_{\ins}, v_{\out})=\omega'(v)$ for all $v\in V'\setminus \{s,t\}$, and weight of all other edges to be zero. Consider any solution $H'\subseteq V'$ of vertex-weighted \blahgeneral. Consider the solution $H$ obtained as follows: for each $s\leadsto t$ path in $H'$ say $s=x_1\rightarrow x_2\rightarrow x_3\rightarrow \ldots \rightarrow x_{r-1}\rightarrow x_{r}=t$ we add the path $x_1\rightarrow x_{2,\ins}\rightarrow x_{2,\out}\rightarrow x_{3,\ins}\rightarrow x_{3,\out}\rightarrow \ldots \rightarrow x_{r-1,\ins}\rightarrow x_{r-1,\out}\rightarrow x_{r}$. Similarly for the $t\leadsto s$ path. Also it is easy to see that for any $v\in V'\setminus \{s,t\}$ we have $\phi(v_{\ins}, v_{\out})=\phi'(v)$, and hence both solutions $H$ and $H'$ have same cost.

\end{proof}
Henceforth we consider only the edge-weighted version of \blahgeneral.

\subsection{Notation}
The set $\{1,2,\ldots, n\}$ is denoted by $[n]$. We denote a directed edge from $u$ to $v$ by $(u,v)$ or $u\rightarrow v$. A directed path from $u$ to $v$ is denoted by $u\leadsto v$. Given a directed graph $G=(V,E)$ the in-degree of $v$ is the number of in-neighbors of $v$ and is denoted by $d^{-}_{G}(v)= |\{w : (w,v)\in E\}|$. Similarly, the out-degree of $v$ is the number of out-neighbors of $v$ and is denoted by $d^{+}_{G}(v)= |\{x : (v,x)\in E\}|$. Given a set $S$ and an integer $r$, the set $S^r$ denotes the set of all $r$-element vectors which have each coordinate from $S$, i.e., $S^{r}= \{(s_1, s_2, \ldots, s_r): s_i\in S\ \forall i\in [r]\}$. Similarly, for $s\in S$ we use $s^r$ to denote the vector $(s_1,s_2,\ldots,s_r)$ where $s_i=s$ for each $i\in [r]$.

\section{An $\pmb{n^{O(k)}}$ algorithm for \pmb{\blah}}
\label{sec:algorithm}

In this section we describe an algorithm for the \blah problem running in $n^{O(k)}$ time where $n$ is the number of vertices in the graph.
First in Section~\ref{subsec:structure} we present a structural property called as \emph{reverse compatibility} for some optimal solution of this problem. Next we define a Token Game in Section~\ref{tokengame} and describe an $n^{O(k)}$ time algorithm for the \textsc{Solving-Token-Game} problem. Finally, in Subsection~\ref{reduction} we present an algorithm that finds the optimum solution of \blah in time $n^{O(k)}$ via a reduction to the \textsc{Solving-Token Game} problem.

\subsection{Structural Lemma for Some Optimal Solution of \pmb{\blah}}
\label{subsec:structure}
\begin{remark}
\label{remark:disjointness-of-forward-paths}
For simplicity, we replace each edge $e$ of the input graph $G$ with $k$ copies $e_1, e_2, \ldots, e_k$, each having the same weight as that of $e$. Let the new graph constructed in this way be $G'$. In $G'$, different $s\leadsto t$ paths must pay each time they use different copies of the same edge. We can alternately view this as the $s\leadsto t$ paths in $G'$ being \textbf{edge-disjoint}.
\end{remark}

\begin{definition}
\label{defn:path-reverse-compatible} \textbf{\emph{(path-reverse-compatible)}}
Let $\fp$ be an $s\leadsto t$ path and $\bp$ be an $t\leadsto s$ path. Let $\{P_1, P_2,\dots,P_d\}$ be the set of maximal sub-paths that $\fp$ and $\bp$ share and for all $j\in [d]$, $P_j$ is the $j$-th sub-path as seen while traversing $\fp$. We say the pair $(\fp,\bp)$ is \emph{path-reverse-compatible} if for all $j\in [d]$, $P_j$ is the $(d-j+1)$-th sub-path that is seen while traversing $\bp$, i.e., $P_j$ is the $j$-th sub-path that is seen while traversing $\bp$ backward.
\end{definition}

\begin{figure}[t]
\centering
\includegraphics[width=3.5in]{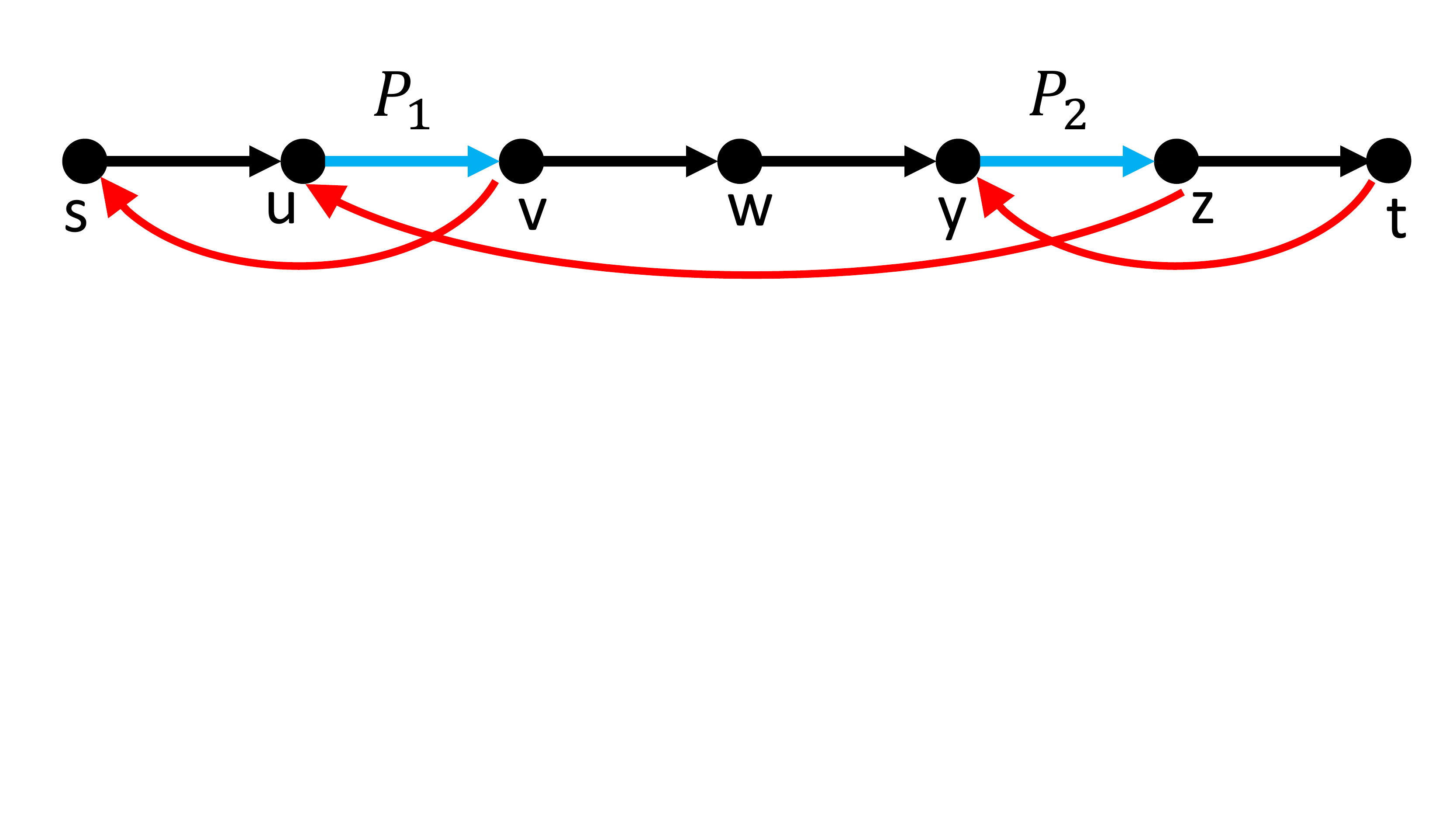}
\vspace{-30mm}
\caption{Let $F$ be an $s\leadsto t$ path given by $s\rightarrow u\rightarrow v\rightarrow w\rightarrow y\rightarrow z\rightarrow t$ and $B$ be an $t\leadsto s$ path given by $t\rightarrow y\rightarrow z\rightarrow u\rightarrow v\rightarrow s$. The two paths $P_1$ and $P_2$ shown in blue are the maximal common sub-paths between $F$ and $B$. From Definition~\ref{defn:path-reverse-compatible}, it follows that $F$ and $B$ are \emph{path-reverse-compatible} since $B$ first sees $P_2$ and then $P_1$. \label{fig:path-new}}
\end{figure}

See Figure~\ref{fig:path-new} for an illustration of path-reverse-compatibility.

\begin{definition}
\label{defn:reverse-compatible} \textbf{\emph{(reverse-compatible)}}
Let $\FP=\{\fp_1,\fp_2,\dots,\fp_r\}$ be a set of $s\leadsto t$ paths and $b$ be an $t\leadsto s$ path. We say $(\FP,\bp)$ is reverse-compatible, if for all $1\leq i \leq r$ the pair $(\fp_i,\bp)$ is path-reverse-compatible.
\end{definition}

The next lemma shows that there exists an optimum solution for \blah which is reverse-compatible.

\begin{lemma}
\label{hosseinlemma} \textbf{\emph{(structural lemma)}}
There exists an optimum solution for \blah which is reverse-compatible.
\end{lemma}
\begin{proof}
In order to prove this lemma, we first introduce the notion of rank of a solution for \blah. Later, we show that an optimum solution of \blah with the minimum rank is reverse-compatible.

\begin{definition} \textbf{\emph{(rank)}}
Let $\FP=\{\fp_1,\fp_2,\dots,\fp_k\}$ be a set of paths from $s\leadsto t$, and $\bp$ be a path from $t\leadsto s$. For each $i\in [k]$, let $d_i$ be the number of maximal sub-paths that $\bp$ and $\fp_i$ share. The \emph{rank} of $(\FP,\bp)$ is given by
$$\rank(\FP,\bp) = \sum_{i=1}^{k} d_i$$
\label{defn:rank}
\end{definition}
Let $(\FP,\bp)$ be an optimum solution of \blah with the minimum rank. Assume for the sake of contradiction that $(\FP,\bp)$ is not reverse-compatible, i.e., there exists some $\fp_i \in \FP$ such that $(\fp_i ,\bp)$ is not path-reverse-compatible. From Definition~\ref{defn:path-reverse-compatible}, this means that $\fp_i$ and $\bp$ share two maximal sub-paths $u\leadsto v$ and $x \leadsto y$, and at the same time $\fp_i$ and $\bp$ both contain $u\leadsto y$ sub-paths (see Figure~\ref{fig:lemma}).

We replace the $u\leadsto y$ sub-path of $\bp$ by the $u\leadsto y$ sub-path of $\fp_i$.
On one hand, $\bp$ shares all of the $u \leadsto y$ sub-path with $\fp_i$. Thus, this change does not increase the weight of the network, therefore it remains an optimum solution. On the other hand, by this change, the sub-paths $u\leadsto v$ and $x\leadsto y$ join. Hence, $d_i$ decreases by 1. Also, since the $s\leadsto t$ paths are edge-disjoint, after the change all other $d_{j}$'s remain same (for $i\neq j$) since $\bp$ shares the whole $u\leadsto y$ sub-path with only $\fp_i$. Therefore, this change strictly decreases the rank of the solution. Existence of an optimum solution with a smaller rank contradicts the selection of $(\FP,\bp)$ and completes the proof.
\end{proof}

\begin{figure}[t]
\centering
\includegraphics[width=3.5in]{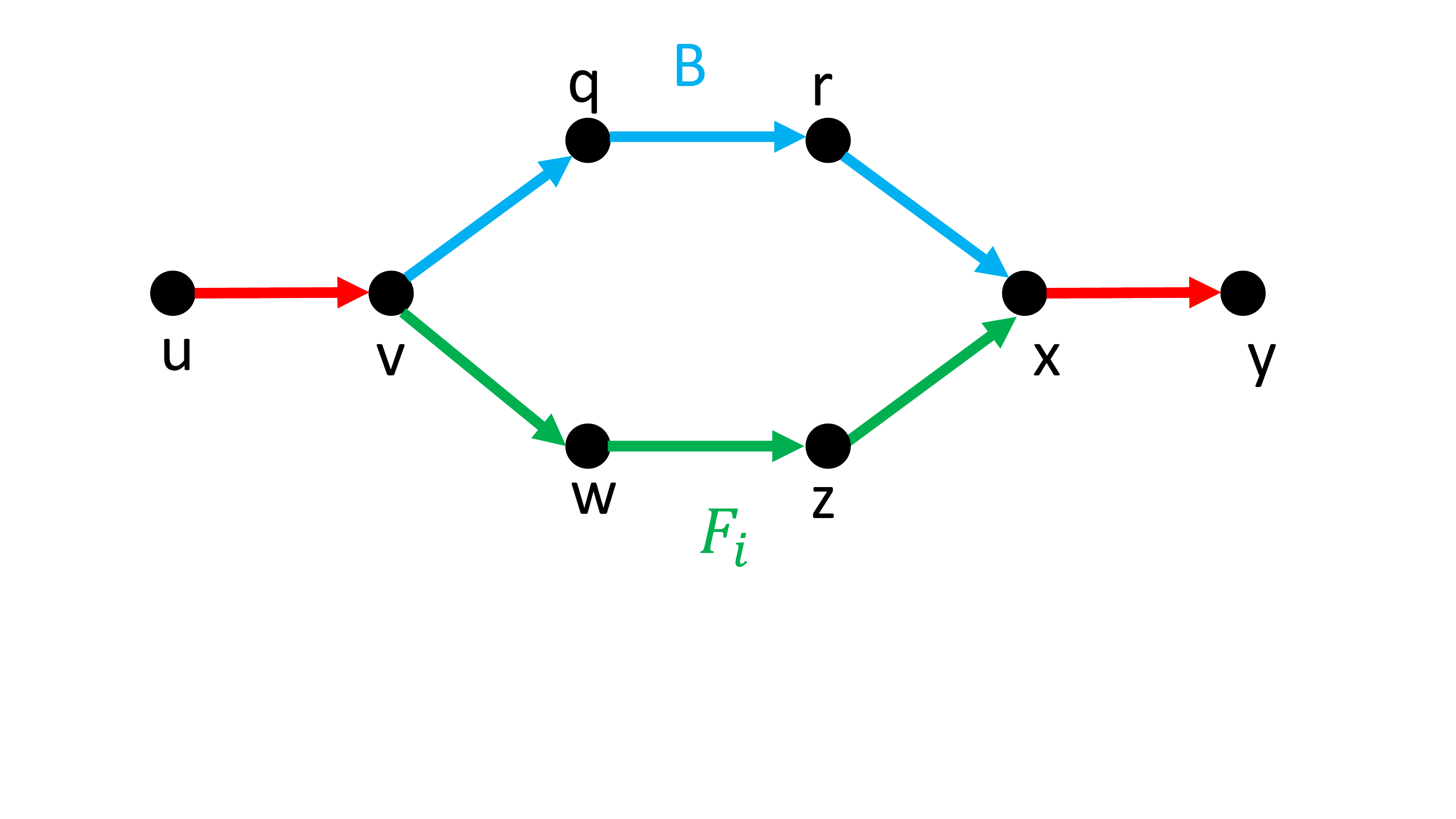}
\vspace{-15mm}
\caption{Let the $u\leadsto y$ sub-path of $F_i$ be a $u\leadsto v\leadsto w\leadsto z\leadsto x\leadsto y$ and the $u\leadsto y$ sub-path of $B$ be $u\leadsto v\leadsto q\leadsto r\leadsto x\leadsto y$. From Definition~\ref{defn:path-reverse-compatible}, it follows that $F_i$ and $B$ are not \emph{path-reverse-compatible} since they both first see $u\leadsto v$ and then see $x\leadsto y$. \label{fig:lemma}}
\end{figure}

\subsection{The Token Game}\label{tokengame}
A Token Game is given by $\langle H, h_1,h_2,\tokens, \moves, \hat{C} \rangle$ where
\begin{itemize}
\item $H=(V_H,E_H)$ is a graph
\item $h_1, h_2$ are two special vertices in $V_H$
\item $\tokens$ is a set of tokens
\item $\moves\subseteq \states\times \states$ is a set of moves
\item $\hat{C}:\moves\rightarrow \mathbb{R}_{\geq 0}$ is a cost function
\end{itemize}

We now define a \emph{state} of the Token Game:
\begin{definition}
A state of the Token Game is an element from the set $V_{H}^{|\tokens|}$, i.e., it is a vector of size $|\tokens|$ where each coordinate comes from $V_H$. It gives the location of each token from $\tokens$.
\end{definition}

A state of the Token Game gives the current location of each token (which vertex of $H$ it is currently on).

The start state is $h_{1}^{|\tokens|}$, i.e.,  when all the tokens are the vertex $h_1$. The end state is $h_{2}^{|\tokens|}$, i.e., when all the tokens are at the vertex $h_2$. The cost function $\hat{C}$ gives the cost of going from one state to another. The goal is to transport all tokens from the start state to the end state with minimum cost. Formally, we have the following problem:

\begin{center}
\noindent\framebox{\begin{minipage}{6.00in}
\textbf{\textsc{Solving-Token-Game}}\\
\emph{\underline{Input} }: A token game $\langle H, h_1, h_2,\tokens, \moves, \hat{C} \rangle$\\
\emph{\underline{Question}}: Find a set of moves of minimum cost to go from the start state $h_{1}^{|\tokens|}$ to the end state $h_{2}^{|\tokens|}$

\end{minipage}}
\end{center}

Now we present an algorithm to solve an instance $\langle H, h_1, h_2,\tokens,\moves,\costtoken \rangle$ of the Token game in time $O(|\moves|+ n^{|\tokens|}\log (n^{|\tokens|}))$
where $n$ is the number of the vertices of $H$.

\begin{lemma}
\textbf{\emph{(algorithm for Token Game)}}
\label{lem:alg-for-token-game}
There exists an algorithm for \textsc{Solving-Token-Game} which runs in time $O(|\moves|+ n^{|\tokens|}\log (n^{|\tokens|}))$.
\end{lemma}
\begin{proof}
%
We build a \emph{game graph} $\hat{H}=(\hat{V}, \hat{E})$ from $H$ as follows: let $\hat{V} = \states$. Recall that $\moves\subseteq \states\times \states$. For each move $M=(\bar{x},\bar{y})\in \moves$ we add an edge $\bar{x}\rightarrow \bar{y}$ in $\hat{H}$ of cost equal to $\hat{C}(M)$.

Note that the starting state $h_{1}^{|\tokens|}$ and the end state $h_{2}^{|\tokens|}$ are both vertices in $\hat{H}$. By the choice of the edges in $\hat{H}$, it is easy to see that finding a shortest $h_{1}^{|\tokens|}\leadsto h_{2}^{|\tokens|}$ path in $\hat{H}$ gives a solution to \textsc{Solving-Token-Game}. We can do this by running Dijkstra's algorithm which takes $O(|E|+|V|\log |V|)$ time on a graph $G=(V,E)$.
In our case $|\hat{V}|= n^{|\tokens|}$ and $|\hat{E}|= |\moves|$. Therefore, we can solve \textsc{Solve-Token-Game} in $O(|\moves|+ n^{|\tokens|}\log (n^{|\tokens|}))$ time.

\end{proof}

\subsection{Reduction from \blah to \textsc{Solving-Token-Game}}\label{reduction}
Here, we provide a reduction from the \blah problem to \textsc{Solving-Token-Game}. As a consequence, we show that one can use the algorithm presented in Subsection~\ref{tokengame} to solve \blah in time $n^{O(k)}$.

Let $I = \langle G=(V,E),s,t, \omega \rangle$ be an instance  of \blah. We now build a Token Game given by $I'=\langle H, h_1, h_2,\tokens, \moves, \hat{C} \rangle$ where $H = G$, $h_1 = s$, $h_2 = t$ and $\tokens=\{\bt,\f_1,\f_2,\ldots,\f_k\}$. Note that $|\tokens|=k+1$. We now describe the set of moves $\moves$ and the associated cost function $\hat{C}$. Fix a state $\bar{v}\in V^{k+1}$, say $\bar{v}=(v_0,v_1, v_2, \ldots, v_{k})$.
\begin{enumerate}
\item \textbf{Backward Move}: For every edge $(w,v_0) \in E(G)$, we add a move $(\bar{v}, \bar{w})$ of cost $\omega(w,v_0)$ where
    \begin{itemize}
    \item $\bar{w}=(w_0, w_1, w_2, \ldots, w_{k})$
    \item $w_0=w$
    \item $w_i=v_i$ for each $i\in [k]$
    \end{itemize}

\item \textbf{Forward Moves}: For every $i\in [k]$ and every edge $(v_i,x) \in E(G)$, we add a move $(\bar{v}, \bar{x})$ of cost $\omega(v_i,x)$ where
    \begin{itemize}
    \item $\bar{x}=(x_0, x_1, x_2, \ldots, x_{k})$
    \item $x_i=x$
    \item $x_j=v_j$ for each $0\leq j\leq k, j\neq i$
    \end{itemize}

\item \textbf{Flip Move}: For each $i\in [k]$ we add a move $(\bar{v}, \bar{y})$ of cost equal to that of the shortest $v_i\leadsto v_0$ path in $G$ where
    \begin{itemize}
    \item $\bar{y}=(y_0, y_1, y_2, \ldots, y_{k})$
    \item $y_0=v_i$
    \item $y_i=v_0$
    \item $y_j=v_j$ for each $0\leq j\leq k, j\notin \{0,i\}$
    \end{itemize}

\end{enumerate}

As in Lemma~\ref{lem:alg-for-token-game}, we build a \emph{game graph} $\hat{G}=(\hat{V}, \hat{E})$ from $G=(V,E)$ as follows: let $\hat{G} = \state$. Recall that $\moves\subseteq \state\times \state$. For each move $M=(\bar{x},\bar{y})\in \moves$ we add an edge $\bar{x}\rightarrow \bar{y}$ in $\hat{H}$ of cost equal to $\hat{C}(M)$.

We now bound the number of moves in $\moves$, which is also equal to the number of edges in $\hat{G}$.

\begin{lemma}
\label{lem:upper-bound-number-of-moves}
The number of moves in $\moves$ (or equivalently the number of edges in $\hat{G}$) is $n^{O(k)}$,
\end{lemma}
\begin{proof}
Fix a state $\bar{v}\in V^{k+1}$, say $\bar{v}=(v_0, v_1, v_2, \ldots, v_{k})$. The number of Backward moves from $\bar{v}$ is $d^{-}_{G}(v_0)$ since we add a backward move for each incoming edge into $v_0$. The number of Forward moves from $\bar{v}$ is $\sum_{i=1}^{k} d^{+}_{G}(v_i)$ since we add a forward move for each outgoing edge from $v_j$ where $j\in [k]$. The number of Flip moves is exactly $k$ since we add a flip move for each $i\in [k]$. Therefore, the degree of $\bar{v}$ in $\hat{G}$ is $d^{-}_{G}(v_0) + \Big(\sum_{i=1}^{k} d^{+}_{G}(v_i) \Big) + k \leq |E|+|E| + n$ since $k\leq n$ and $\sum_{v\in V} d^{+}_{G}(v) = |E|= \sum_{v\in V} d^{-}_{G}(v)$. Hence, the max degree of $\hat{G}$ is $2|E|+n$. Since $|\hat{V}|=n^{k+1}$ it follows that $|\moves|=|\hat{E}|\leq n^{k+1}\cdot (2|E|+n) = n^{O(k)}$ as $|E|= O(n^{2})$
\end{proof}

Next we show that $\opt(I) = \opt(I')$, where $\opt(I)$ and $\opt(I')$ denote for the optimum solutions of $I$ and $I'$ respectively. We do this by the following two lemmas:

\begin{lemma}\label{lm1}
$\opt(I) \leq \opt(I')$.
\end{lemma}
\begin{proof}
Let $\s'$ be a solution of $I'$ of cost $\opt(I')$. Then $\s'$ is a shortest $s^{k+1}\leadsto t^{k+1}$ path in $\hat{G}$. Each edge in $\hat{G}$ corresponds to a move from $\moves$. Let the moves corresponding to the path $\s'$ be $M_1, M_2, \ldots, M_r$.
%
%
%
Consider a move $M\in \{M_1. M_2. \ldots, M_r\}$ and say $M=(\bar{v}, \bar{w})$ where $\bar{v}=(v_0, v_1, v_{2}, \ldots, v_{k})$ and $\bar{w}=(w_0, w_1, w_2, \ldots, w_{k})$. We now build a solution $\s$ for \blah as follows: there are three cases to consider depending on the type of the move $M$
\begin{itemize}
\item \textbf{$M$ is a Backward Move}: Then we add the edge $(w_0, v_0)$ to $\s$. Note that $\hat{C}(M) = \omega(w_0, v_0)$.

\item \textbf{$M$ is a Forward Move, say for token $\f_i$}: Then we add the edge $(v_{i}, w_{i})$ to $\s$. Note that $\hat{C}(M)=\omega(v_{i}, w_{i})$.

\item \textbf{$M$ is a Flip Move, say between token $\f_i$ and $\bt$}: Let $P$ be the shortest $v_i\leadsto v_0$ path. Then we add the path $P$ to $\s$. Note that $\hat{C}(M) =\sum_{e\in P} \omega(e)$.
\end{itemize}

Since $\s'$ is a solution of $I'$ it follows that $\s$ is indeed a solution of $I$: the path traced by the moves of backward token $\bt$ gives an $t\leadsto s$ path and the paths traced by moves of the $k$ forward tokens $\f_1, \f_2, \ldots, f_k$ give the $k$ different $s\leadsto t$ paths . Also the cost of $\s'$ is equal to $\sum_{i=1}^{r} \hat{C}(M_i)$. As we have seen above, we were able to construct a solution $\s$ for $I$ from $\s'$. Note that for each edge $e$, its contribution to cost of $\s'$ is $\omega(e)$ which is greater than or equal to its contribution to cost of $\s$, since there might be some savings due to sharing of edges between the $t\leadsto s$ path and some $s\leadsto t$ path\footnote{Consider a path $x\rightarrow y\rightarrow z$. Let $e_1 = (x,y)$ and $e_2=(y,z)$. Suppose token $\f_1$ is at $x$ and token $\bt$ is at $z$ and they want to exchange positions. A flip move would result in cost equal to $\omega(e_1)+\omega(e_2)$. However, we can have a move sequence of higher cost which results in same final positions for $\f_1$ and $\bt$ as follows: first $\f_1$ makes a forward move and reaches $y$ with cost $\omega(e_1)$. Then there is a flip move of cost $\omega(e_2)$ which brings $\bt$ to $y$ and takes $\f_1$ to z. Finally $\bt$ makes a backward move of cost $\omega(e_1)$ to reach $x$. The total cost of this move sequence is $2\omega(e_1) +\omega(e_2)$, which is greater than the original cost of $\omega(e_1)+\omega(e_2)$}. Thus the cost of $\s$ is at most the cost of $\s'$. Given any solution $\s'$ for $I'$, we were able to construct a solution, say $\s$, for $I$ of cost less than or equal to that of $\s'$. Therefore, it follows that $\opt(I) \leq \opt(I')$.

\end{proof}

\begin{lemma}\label{lm2}
$\opt(I) \geq \opt(I')$ .
\end{lemma}
\begin{proof}
To prove this lemma, we use Lemma~\ref{hosseinlemma} which states there exists an optimal solution, say $\s$ , for $I$ which is reverse-compatible. We now build a solution $\s'$ for $I'$ which has the same cost as that of $\s$. This is sufficient to prove the lemma.

As observed in Remark~\ref{remark:disjointness-of-forward-paths}, we can assume that the $s\leadsto t$ paths are pairwise edge-disjoint. Let the $t\leadsto s$ path in $\s$ be $Q$ and the $k$ paths from $s\leadsto t$ be $P_1, P_2, \ldots, P_k$. For $i\in [k]$, let $P_{i,1}, P_{i,2}, \ldots, P_{i,r_i}$ be the maximal sub-paths shared between $P_i$ and $Q$ as seen in order from $P_i$. The reverse-compatibility of $\s$ implies that if we  traverse $Q$ backwards then we again see the paths in the same order, namely $P_{i,1}, P_{i,2} \ldots, P_{i,r_i}$. Let us call these paths which are shared between $Q$ and an $s\leadsto t$ path as \textbf{common-paths}.

We build $\s'$ by adding moves as follows: for each $i\in [k]$, add Forward moves (by selecting shortest paths) to transport each $\f_i$ from $s$ to the starting point of $P_{i,1}$. Follow $Q$ in backwards direction as it travels from $s$ to $t$. We move $\bt$ backwards along $Q$ until it reaches end-point $y$ of some common-path say $x\leadsto y$. Since we only require one $t\leadsto s$ path it follows that there exists unique $j\in [k]$ such that $P_{j,1}=x\leadsto y$. Then we add a Flip move between $\f_j$ (which is located at $x$) and $\bt$ (which is located at $y$). Continue to follow $Q$ backwards and move $\bt$ along it by Backward moves until either $\bt$ reaches $t$ or $\bt$ reaches end-point of another common path. If $\bt$ reaches end-point of another common path then we again do a Flip move as above. Otherwise if $\bt$ reaches $t$, then each forward token $\f_i$ is at the end-point of $P_{i,r_i}$. Add Forward moves (by selecting shortest paths) to transport the forward tokens from their current locations to $t$. Let the final set of moves be the solution $\s'$.

$\s$ is a valid solution for $I$ implies that we have $k$ paths from $s\leadsto t$ and one path from $t\leadsto s$. So, the moves we add indeed take all the $(k+1)$ tokens from the start state $s^{k+1}$ to the end state $t^{k+1}$, i.e., $\s'$ is a valid solution for $I'$.
We now show that the cost $\s'$ is equal to that of $\s$.
Let $e$ be any edge in $\s$. If $e$ is not part of a common-path, then in $\s'$ we only pay for it once in either a Backward move or a Forward move. On the other hand, if $e$ is part of a common path then in $\s'$ again we also pay for it only once in a Flip move. Therefore, cost of $\s'$ is equal to cost of $\s$ which implies that $\opt(I)\geq \opt(I')$.

\end{proof}

\begin{theorem}
\label{thm:blah-time}
There exists an algorithm that solves \blah in time $n^{O(k)}$, where $n$ is the number of vertices of the input graph.
\end{theorem}
\begin{proof}
Let $I=\langle G,s,t,\omega \rangle$ be an instance of \blah. As described in Section~\ref{reduction}, we build an instance $I'=\langle G, s, t,\tokens=(\bt, \f_1, \f_2, \ldots, \f_k), \moves, \hat{C} \rangle$ of \textsc{Solving-Token-Game}. Lemmas \ref{lm1} and \ref{lm2} imply that $I$ can be solved by instead solving the instance $I'$. By Lemma~\ref{lem:upper-bound-number-of-moves}, the number of moves in $I'$ is $|\moves|=n^{O(k)}$.
By Lemma \ref{lem:alg-for-token-game} we can solve $I'$ in time $O(|\moves|+ n^{|\tokens|}\log (n^{|\tokens|}) = O(n^{O(k)}+ n^{k+1}\log (n^{k+1})) = n^{O(k)}$
%
%
\end{proof}

\subsection{Structural Lemma fails for \pmb{\blahgeneral} when $\pmb{\min\{k_1, k_2\}\geq 2}$}
\label{app:no-structure}

Recall that in the \blahgeneral problem we want $k_1$ paths from $s\leadsto t$ and $k_2$ paths from $t\leadsto s$. So, we define a natural extension of Definition~\ref{defn:path-reverse-compatible} to reverse-compatibility of a set of forward paths and a set of backward paths as follows.
\begin{definition} \textbf{\emph{(general-reverse-compatible)}}
\label{defn:reverse-compatible-general}
Let $\FP=\{\fp_1,\fp_2,\dots,\fp_{k_1}\}$ be a set of $s\leadsto t$ paths and $\BP=\{\bp_1,\bp_2,\dots,\bp_{{k_2}}\}$ be a set of $t\leadsto s$ paths. We say $(\FP,\BP)$ is general-reverse-compatible, if for all $1\leq i \leq {k_2}$, $(\FP,\bp_i)$ is reverse-compatible.
\end{definition}

\begin{figure}[t]
\centering
\includegraphics[width=6in]{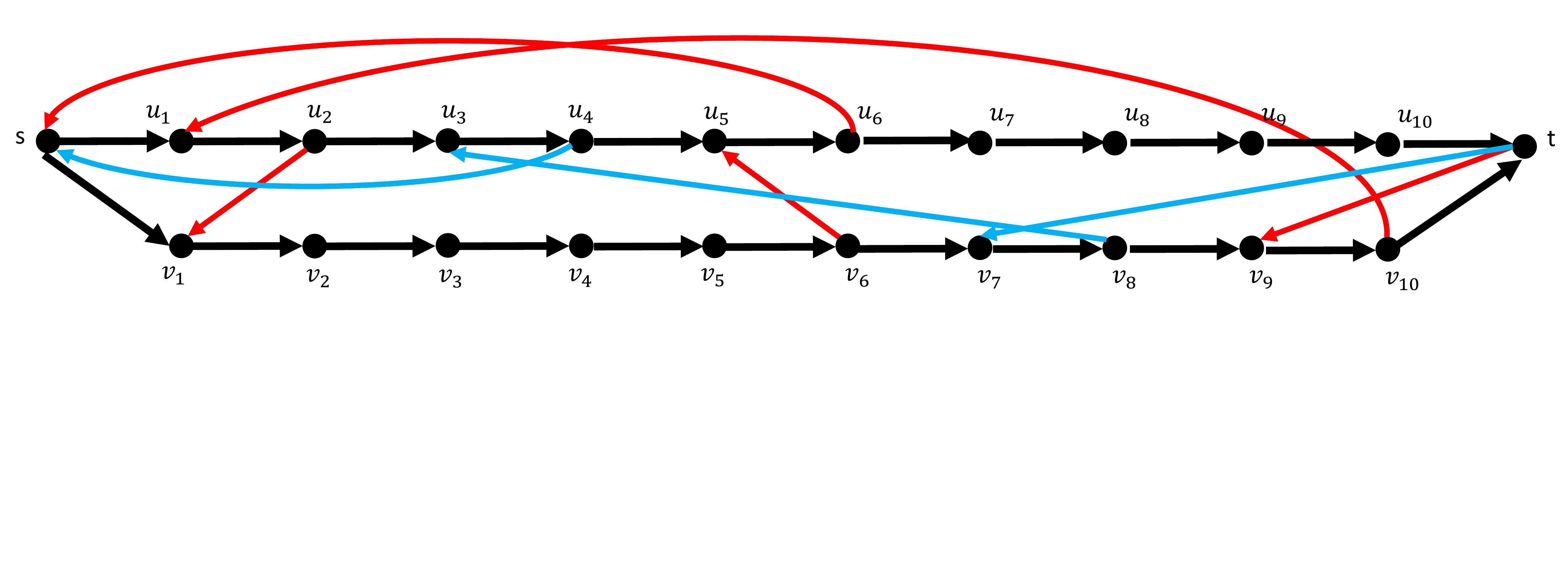}
\vspace{-30mm}
\caption{Each black edge has weight 1, each red edge and each blue edge has weight 0. \label{fig:counter}}
\end{figure}

The following theorem shows that Lemma~\ref{hosseinlemma} does not hold for the  \blahgeneral problem when $\min\{k_1, k_2\}\geq 2$, i.e.,~Lemma \ref{hosseinlemma} is in its most general form.
\begin{theorem}
There exists an instance of $2$-SCSS-$(2,2)$ in which no optimum solution is general-reverse-compatible.
\end{theorem}
\begin{proof}
Figure \ref{fig:counter} illustrates an example of the 2-SCSS(2,2) problem in which no optimal solution satisfies the reverse compatibility condition. Let the weight of the black edges be 1, and weight of all the other edges be 0. Since we have edges of weight 0, we will henceforth only consider the paths which do not have vertices repeating.

Let $P_1$ be the path $s\rightarrow u_1\rightarrow u_2\rightarrow \ldots \rightarrow u_9\rightarrow u_{10}\rightarrow t$ and $P_2$ be the path $s\rightarrow v_1\rightarrow v_2\rightarrow \ldots \rightarrow v_9\rightarrow v_{10}\rightarrow t$. Note that $P_1$ and $P_2$ are edge-disjoint and have weight 11 each. We now give a solution of total weight 22: take $P_1$ and $P_2$ as the two $s\leadsto t$ paths. For the two $t\leadsto s$ paths take $P_{3}:=t\rightarrow v_7\rightarrow v_8\rightarrow u_3\rightarrow u_4\rightarrow s$ and $P_{4}:= t\rightarrow v_9\rightarrow v_{10}\rightarrow u_1\rightarrow u_2\rightarrow v_1 \rightarrow v_2\rightarrow v_3\rightarrow v_4\rightarrow v_5\rightarrow v_6\rightarrow u_5\rightarrow u_6\rightarrow s$. Since every black edge is used exactly once in the outgoing paths and at most once in the incoming paths, it is easy to verify that the total weight of this solution is 22. Moreover, this solution is not general-reverse-compatible since the paths $P_1$  and $P_4$ do not satisfy the path-reverse-compatibility condition (recall Definition~\ref{defn:path-reverse-compatible}).

Therefore, to prove the theorem, it is now enough to show that all other solutions have a weight at least 23. A simple observation is that any solution has weight at least 22 since the shortest path from $s$ to $t$ has weight 11. Moreover, there are exactly two such $s\leadsto t$ paths of weight 11, viz. $P_1$ and $P_2$. Hence suppose to the contrary that there is a solution, say $\textbf{S}$, of weight exactly 22. We now show that $\textbf{S}$ must exactly be the solution described in above paragraph. We first show the following lemma:
\begin{lemma}
Any $t\leadsto s$ path uses at least one black edge from each of $P_1$ and $P_2$.
\end{lemma}
\begin{proof}
Note that there are only two edges outgoing from $t$: a blue edge and a red edge.  Suppose the first edge on $t\leadsto s$ path is the red edge $t\rightarrow v_9$. Then we must reach $v_{10}$ since the only outgoing edge from $v_9$ is $v_9\rightarrow v_{10}$. From $v_{10}$, we can either go back to $t$ (and start the argument again) or the other option is to go to $u_1$ which forces the use of edge $u_1\rightarrow u_2$. So we have used $v_9\rightarrow v_{10}$ from $P_2$ and $u_1\rightarrow u_2$ from $P_1$.

Suppose the first edge on $t\leadsto s$ path is the blue edge $t\rightarrow v_7$. This forces the use of the edge $v_7\rightarrow v_8$ from $P_2$ since  it is the only outgoing edge from $v_7$. From $v_8$, we can either reach $v_{9}$ (and the same argument applies as in previous case) or $u_3$. Reaching $u_3$ forces the use of the edge $u_3\rightarrow u_4$ from $P_1$ since it is the only outgoing edge from $u_3$.
\end{proof}
Hence, in order to obtain a solution of weight exactly 22 we cannot take either $P_1$ twice or $P_2$ twice for the choice of the two $s\leadsto t$ paths: since this itself gives a weight of 22, and the above claim implies a weight of at least 1 from the ``other" path. This shows the correctness of the following lemma:
\begin{lemma}
\label{lem:one-each}
The two $s\leadsto t$ paths in $\textbf{S}$ are exactly $P_1$ and $P_2$. Hence, to maintain a weight of exactly 22 it follows that we cannot use any black edge twice in the $t\leadsto s$ paths in $\textbf{S}$.
\end{lemma}

Observe that we still need to choose two $t\leadsto s$ paths, say $Q_1$ and $Q_2$, in $\textbf{S}$. The following lemma shows that $\textbf{S}$ needs to use both the red edge and blue edge outgoing from $t$:
\begin{lemma}
\label{lem:q1-q2}
Without loss of generality, the first edges of $Q_1$ and $Q_2$ are $t\rightarrow v_7$ and $t\rightarrow v_9$
\end{lemma}
\begin{proof}
Suppose not. Since the only two outgoing edges from $t$ are the blue edge $t\rightarrow v_7$ and the red edge $t\rightarrow v_9$, it follows that the first edge of both $Q_1$ and $Q_2$ is the same (and is either $t\rightarrow v_7$ or $t\rightarrow v_9$). Suppose the first edge of both $Q_1$ and $Q_2$ is $t\rightarrow v_7$ (the argument for the first edge being $t\rightarrow v_9$ is similar). Since $v_7\rightarrow v_8$ is the only outgoing edge from $v_7$, this implies that we must choose this edge in both $Q_1$ and $Q_2$. Since the two $s\leadsto t$ paths in $\textbf{S}$ are $P_1$ and $P_2$, this shows that the weight of $\textbf{S}$ is at least 23.
\end{proof}

Let us now consider the path $Q_1$: it starts with the edge $t\rightarrow v_7$. Since the only outgoing edges from $v_7, u_3$ are $v_7\rightarrow v_8, u_3\rightarrow u_4$ respectively it follows that $Q_1$ contains the sub-path $Q'_{1}:= t\rightarrow v_7\rightarrow v_8\rightarrow u_3\rightarrow u_4$. Similarly for $Q_2$, the first edge being $t\rightarrow v_9$ implies that it contains the sub-path $Q'_{2}:= t\rightarrow v_9\rightarrow v_{10}\rightarrow u_1\rightarrow u_2$. After this, $Q_2$ cannot contain the edge $u_2\rightarrow u_3$ (since this would force it to also use the edge $u_3\rightarrow u_4$, which was already used by $Q_1$). Hence after $Q'_{2}$, the path $Q_2$ must follow the sub-path $u_2\rightarrow v_1\rightarrow v_{2}\rightarrow v_{3}\rightarrow v_{4}\rightarrow v_{5}\rightarrow v_{6}$. After reaching $v_{6}$, the path $Q_2$ has two choices: either use the edge $v_6\rightarrow v_7$ or $v_6\rightarrow v_5$. But it cannot use the edge $v_{6}\rightarrow v_{7}$ since that would force it to use the edge $v_{7}\rightarrow v_{8}$, which was already used by $Q_1$. Therefore, from $v_6$ the path $Q_2$ reaches $u_5$ and is then forced to reach $u_6$. At this point $Q_2$ has two choices: either continue from $u_6$ to $t$ (in which case we again apply the whole argument starting from Lemma~\ref{lem:q1-q2}), or use the edge $u_6\rightarrow s$ of weight 0. Therefore we have that $Q_2$ is exactly the path $P_{4}:= t\rightarrow v_9\rightarrow v_{10}\rightarrow u_1\rightarrow u_2\rightarrow v_1 \rightarrow v_2\rightarrow v_3\rightarrow v_4\rightarrow v_5\rightarrow v_6\rightarrow u_5\rightarrow u_6\rightarrow s$. It remains to show that the path $Q_1$ is exactly $P_3$. We know that $Q_1$ contains the sub-path $Q'_{1}:= t\rightarrow v_7\rightarrow v_8\rightarrow u_3\rightarrow u_4$. From $u_4$, there are two choices: either use the edge $u_4\rightarrow s$ of weight 0, or use the edge $u_4\rightarrow u_5$. However, in the second choice, the next edge on $Q_2$ must be $u_5\rightarrow u_6$. But this edge was already used by $Q_2$ which contradicts Lemma~\ref{lem:q1-q2}. This shows that $Q_1$ is exactly the path $P_{3} =t\rightarrow v_7\rightarrow v_8\rightarrow u_3\rightarrow u_4\rightarrow s$, which completes the proof of the theorem.
\end{proof}

\section{$\pmb{f(k)\cdot n^{o(k)}}$ Hardness for \pmb{\blah}}
\label{sec:hardness}

In this section we prove Theorem~\ref{thm:main-hardness}. 
We reduce from the \textsc{Grid Tiling} problem (see Section~\ref{subsec:our-results} for the definition).
Chen et al.~\cite{chen-hardness} showed that for any computable function $f$, the existence of an $f(k)\cdot n^{o(k)}$ algorithm for $k$-\textsc{Clique} implies
ETH fails.
Marx~\cite{daniel-grid-tiling} gave the following reduction which transforms the problem of finding a $k$-\textsc{Clique} into an
instance of $k\times k$ \gt as follows: For a graph $G=(V, E)$ with $V=\{v_1, v_2, \ldots, v_n\}$ we build an instance $I_G$ of \gt
\begin{itemize}
\item For each $1\leq i\leq k$, we have $(j,\ell)\in S_{i,i}$ if and only if $j=\ell$.
\item For any $1\leq i\neq j\leq k$, we have $(\ell, r)\in S_{i,j}$ if and only if $\{v_{\ell}, v_r\}\in E$.
\end{itemize}
It is easy to show that $G$ has a clique of size $k$ if and only if the instance $I_G$ of \gt has a solution. Therefore, assuming ETH, even the following special case of $k\times
k$ \gt also cannot be solved in time $f(k)\cdot n^{o(k)}$ for any computable function $f$:

\begin{center}
\noindent\framebox{\begin{minipage}{4.75in}
\textbf{\textsc{$k\times k$ Grid Tiling}*}\\
\emph{Input }: Integers $k, n$, and $k^2$ non-empty sets $S_{i,j}\subseteq [n]\times [n]$ where $1\leq i, j\leq k$ such that for each $1\leq i\leq k$, we have $(j,\ell)\in S_{i,i}$ if and only if $j=\ell$\\
\emph{Question}: For each $1\leq i, j\leq k$ does there exist a value $\gamma_{i,j}\in S_{i,j}$ such that
\begin{itemize}
\item If $\gamma_{i,j}=(x,y)$ and $\gamma_{i,j+1}=(x',y')$ then $x=x'$.
\item If $\gamma_{i,j}=(x,y)$ and $\gamma_{i+1,j}=(x',y')$ then $y=y'$.
\end{itemize}
\end{minipage}}
\end{center}

We actually give a reduction from this problem to \blah. To the best of our knowledge, this is the first use of the special structure of \textsc{Grid Tiling}* in a W[1]-hardness proof.

 \begin{figure}[t]
\flushleft
 \includegraphics[width=6.7in]{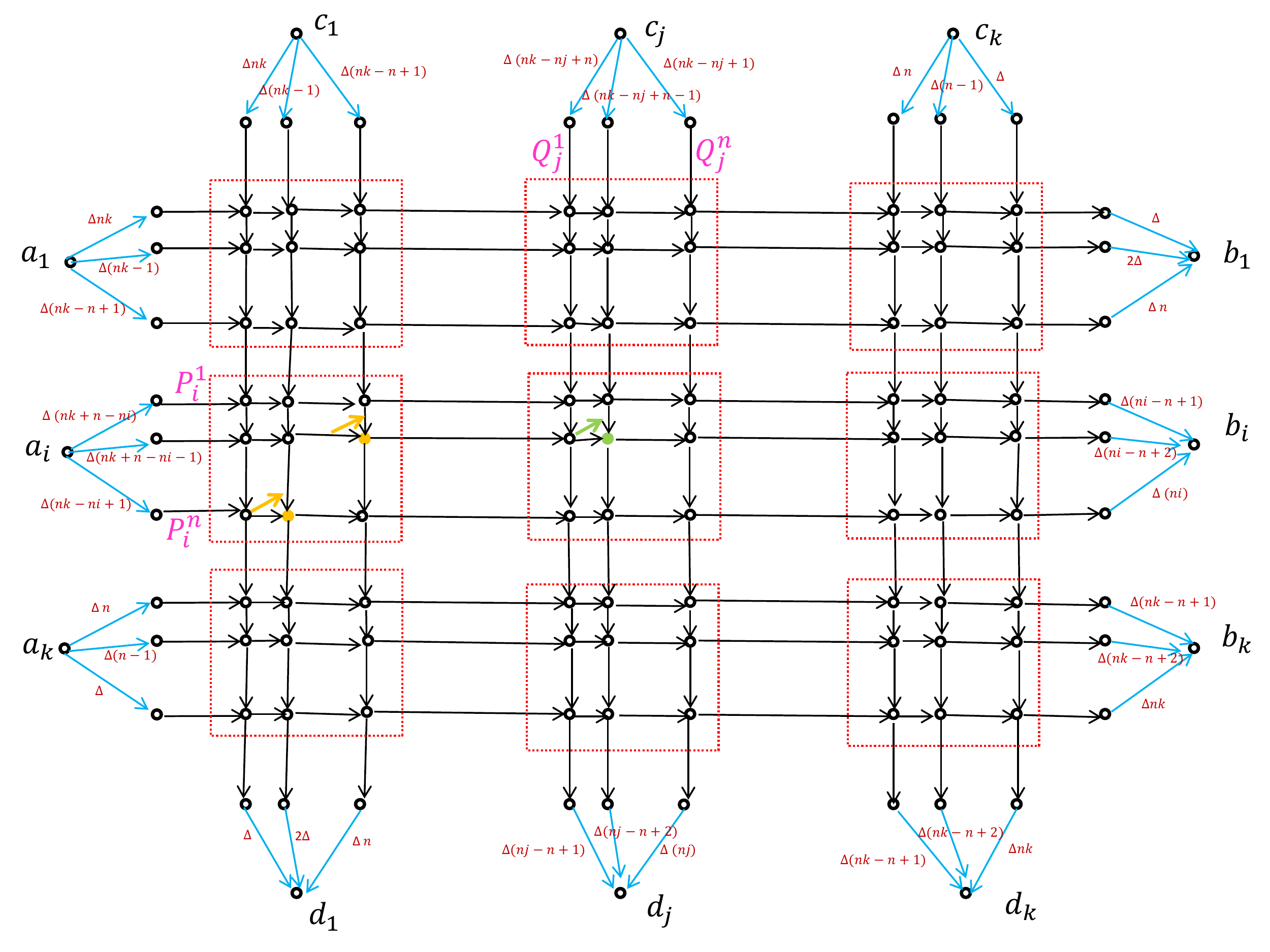}
 \caption{The instance of \blaha created from an instance of Grid Tiling*.
 \label{fig:dsf}}
 \end{figure}

Consider an instance of \textsc{Grid Tiling}*. We now build an instance of edge-weighted \blaha as shown in Figure~\ref{fig:dsf}.
We consider $4k$ special vertices: $(a_i, b_i, c_i, d_i)$ for each $i\in [k]$. We introduce $k^2$ red gadgets where
each gadget is an $n\times n$ grid. Let weight of each black edge be $4$.

\begin{definition}
For each $1\leq i\leq k$, an $a_i \leadsto b_i$ \emph{canonical} path is a path from $a_i$ to $b_i$ which starts with a blue edge coming out of $a_i$,
then follows a horizontal path of black edges and finally ends with a blue edge going into $b_i$. Similarly a $c_j\leadsto
d_j$ \emph{canonical} path is a path from $c_j$ to $d_j$ which starts with a blue edge coming out of $c_j$, then follows a
vertically downward path of black edges and finally ends with a blue edge going into $d_j$.
\end{definition}

For each $1\leq i\leq k$, there are $n$ edge-disjoint $a_i \leadsto b_i$ canonical paths: let us call them $P^{1}_{i}, P^{2}_{i}, \ldots, P^{n}_i$ as
viewed from top to bottom. They are named using magenta color in Figure~\ref{fig:dsf}. Similarly we call the canonical paths
from $c_j$ to $d_j$ as $Q^{1}_{j}, Q^{2}_{j}, \ldots, Q^{n}_j$ when viewed from left to right. For each $i\in [k]$ and
$\ell\in [n]$ we assign a weight of $\Delta(nk-ni+n+1-\ell), \Delta(ni-n+\ell)$ to the first, last edges of $P^{\ell}_{i}$ (which are colored blue)
respectively. Similarly for each $j\in [k]$ and $\ell\in [n]$ we assign a weight of $\Delta(nk-nj+n+1-\ell), \Delta(nj-n+\ell)$ to the
first, last edges of $Q^{\ell}_{j}$ (which are colored blue) respectively. Thus the total weight of first and last blue edges on any canonical path
is exactly $\Delta(nk+1)$. The idea is to choose $\Delta$ large enough such that in any optimum solution the paths between the
terminals will be exactly the canonical paths. We will see that $\Delta = 7n^{6}$ will suffice for our reduction. Any canonical path
uses two blue edges (which sum up to $\Delta(nk+1)$), $(k+1)$ black edges not inside the red gadgets and $(n-1)$ black edges inside
each gadget. Since the number of gadgets that each canonical path visits is $k$ and the weight of each black edge is 4, it follows that the
total weight of any canonical path is
\begin{equation}
\label{eqn:alpha}
\alpha = \Delta(nk+1)+4(k+1)+4k(n-1)
\end{equation}

\begin{figure}
 \centering
 \includegraphics[height=3in]{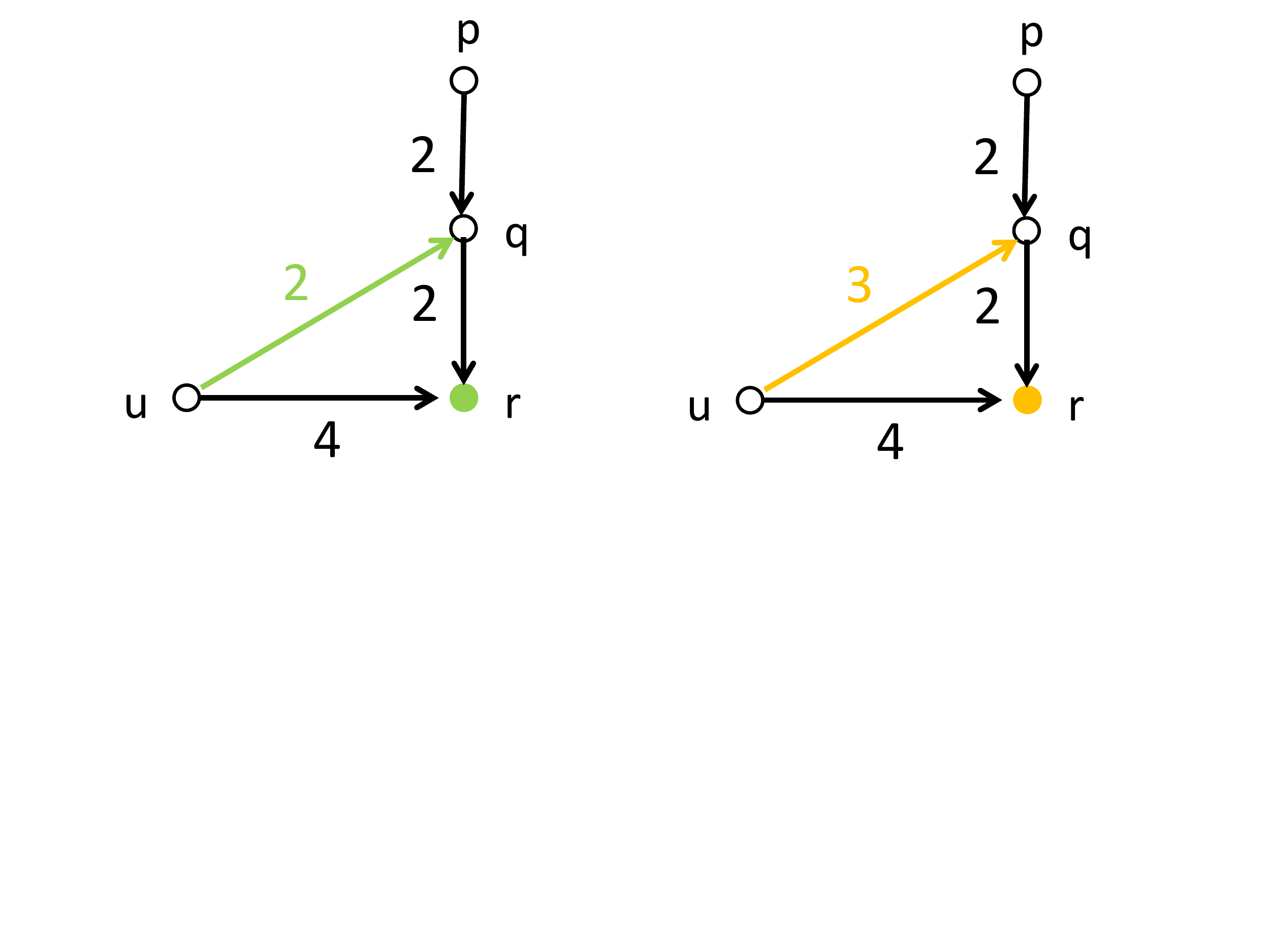}
 \vspace{-40mm}
 \caption{Let $u, r$ be two consecutive vertices on the canonical path say $P^{\ell}_{i}$. Let $r$ be on the canonical path
 $Q^{\ell'}_{j}$ and let $p$ be the vertex preceding it on this path. If $r$ is a green (respectively orange) vertex then we subdivide the edge $(p, r)$ by introducing a new vertex $q$
 and adding two edges $(p, q)$ and $(q, r)$ of weight 2. We also add an edge $(u, q)$ of weight 2 (respectively 3). The idea is if both the edges $(p, r)$ and $(u, r)$ were
 being used initially then now we can save a weight of 2 (respectively 1) by making the horizontal path choose $(u, q)$ and then we get $(q, r)$ for free, as it is already being used
 by the vertical canonical path.
 \label{fig:savings-new}}
 \end{figure}

Intuitively the $k^2$ red gadgets correspond to the $k^2$ sets in the \textsc{Grid Tiling}* instance. Let us denote the gadget
which is the intersection of the $a_i \leadsto b_i$ paths and $c_j \leadsto d_j$ paths by $G_{i,j}$. If $i=j$, then we call $G_{i,j}$ as a symmetric gadget; else we call it as an asymmetric gadget. We perform the following modifications on the edges inside the gadget: (see Figure~\ref{fig:dsf})
\begin{itemize}
\item \textbf{Symmetric Gadgets}: For each $i\in [k], x\in [n]$ we color green the vertex in the gadget $G_{i,i}$ which is the unique intersection of the canonical paths $P_{i}^{x}$ and $Q_{i}^{x}$. Then we add a shortcut as shown in Figure~\ref{fig:savings-new}. The idea is that we will enforce that the $a_i \leadsto b_i$ path is used as part of the $t\leadsto s$ path and the $c_i\leadsto d_i$ path is used as part of one of the $(2k-1)$ $s\leadsto t$ paths. Hence, if both the $a_i \leadsto b_i$ path and $c_i \leadsto d_i$ path pass through the green vertex then the $a_i \leadsto b_i$ path can save a weight of 2 by using the green edge (which has weight 2) and a vertical downward edge (which is free since already being used by $c_j \leadsto d_j$ canonical path) to reach the green vertex, instead of paying a weight of 4 to use the horizontal edge reaching the green vertex.

\item \textbf{Asymmetric Gadgets}: For each $i\neq j\in [k]$, if $(x,y)\in S_{i,j}$ then we color orange the vertex in the gadget $G_{i,j}$ which is the unique intersection of the canonical paths $P_{i}^{x}$ and $Q_{j}^{y}$. Then we add a shortcut as shown in Figure~\ref{fig:savings-new}. Similar to above, the idea is if both the $a_i \leadsto b_i$ path and $c_j \leadsto d_j$ path pass through the orange vertex then the $a_i \leadsto b_i$ path can save a weight of 1 by using the orange edge of weight 3 followed by a vertical downward edge (which is free since already being used by the $c_j \leadsto d_j$ canonical path) to reach the orange vertex, instead of paying a weight of 4 to use the horizontal edge reaching the orange vertex.
\end{itemize}

\subsection{Vertices and Edges not shown in Figure~\ref{fig:dsf}}
The following vertices and edges are not shown in Figure~\ref{fig:dsf} for sake of presentation:
\begin{itemize}
\item Add two vertices $s$ and $t$.
\item For each $1\leq i\leq k$, add an edge $(s, c_i)$ of weight $0$.
\item For each $1\leq i\leq k$, add an edge $(d_i, t)$ of weight $0$
\item Add edges $(t, a_k)$ and $(b_1, s)$ of weight 0.
\item For each $2\leq i\leq k$, introduce two new vertices $e_i$ and $f_i$.
\item For each $2\leq i\leq k$, add edges $s\rightarrow e_i$ and $f_i\rightarrow t$ of weight $0$.
\item For each $2\leq i\leq k$, add a path $b_{i}\rightarrow e_{i}\rightarrow f_{i}\rightarrow a_{i-1}$. Set the weights of $(b_i, e_i)$ and $(f_{i}, a_{i-1})$ to be $0$.
\item For each $2\leq i\leq k$, set the weight of the edge $(e_i, f_i)$ to be $W$. We call these edges as \textbf{connector} edges.  The idea is that we will choose $W$ large enough so that each connector edge is used exactly once in an optimum solution for \blaha. We will see later that $W=53n^{9}$ suffices for our reduction.
\end{itemize}

\begin{remark} \emph{We need a small technical modification: add one dummy row and column to the \textsc{Grid Tiling}* instance. Essentially, we
now have a dummy index $1$. So neither the first row nor the first column of any $S_{i,j}$ has any elements in the
\textsc{Grid Tiling}* instance. That is, no green vertex or orange vertex can be in the first row or first column of any gadget.}
\end{remark}

We now prove two theorems which together give a reduction from \textsc{Grid Tiling}* to \blaha. Let
\begin{equation}
\beta = 2k\cdot \alpha + W(k-1) - (k^2 + k)
\label{eqn:beta}
\end{equation}

\subsection{\textsc{Grid Tiling}* has a solution $\pmb{\Rightarrow}$ \pmb{\blaha} has a solution of weight $\pmb{\leq \beta}$}
\label{thm:easy-direction}

First we show the easy direction.

\begin{theorem}
\label{thm:dsf-redn-easy} If the \textsc{Grid Tiling}* instance has a solution then \blaha has a solution of weight at most $\beta$.
\end{theorem}
\begin{proof}
For each $1\leq i,j\leq k$ let $s_{i,j}\in S_{i,j}$ be the element in the solution of the \textsc{Grid Tiling}* instance.
Therefore, there exist $k$ numbers $\delta_1, \delta_2, \ldots, \delta_k$ such that $s_{i,j}=(\delta_i, \delta_j)$ for each $1\leq i,j\leq k$. We use the following path for the $t\leadsto s$ path in our solution:
\begin{itemize}
\item First use the edge $(t, a_k)$. This incurs weight 0.
\item For each $k\geq i\geq 2$, use the canonical $a_i \leadsto b_i$ path $P^{\delta_i}_{i}$ followed by the path $b_{i}\rightarrow e_i \rightarrow f_i\rightarrow a_{i-1}$. This way we reach $a_1$. Finally use the canonical path $P^{\delta_1}_{1}$ to reach $b_1$. The total weight of these edges is $\alpha\cdot k + W(k-1)$ since we have used $k$ canonical paths and $(k-1)$ connector edges.
\item Finally use the edge $(b_1, s)$ of weight 0.
\end{itemize}
Therefore, with a total weight of $\alpha\cdot k + W(k-1)$ we have obtained an $t\leadsto s$ path. Since we have used all the $k-1$ connector edges in the $t\leadsto s$ path, we can now use them for free in (different) $s\leadsto t$ paths. In particular, we get $(k-1)$ $s\leadsto t$ paths given by $s\rightarrow e_i \rightarrow f_i \rightarrow t$ for each $2\leq i\leq k$. Note that the total weight of these $(k-1)$ $s\leadsto t$ paths is 0, since for each $2\leq i\leq k$ the edge $(e_i, f_i)$ is obtained for free (as it was used in the $t\leadsto s$ path) and both the edges $(s, e_i)$ and $(f_i, t)$ have weight 0.

Note that we still need $k$ more $s\leadsto t$ paths. For each $j\in [k]$, we add the canonical $c_j \leadsto d_j$ path $Q^{\delta_j}_{j}$. For each $j\in [k]$, note that the edges $(s, c_j)$ and $(d_j, t)$ have weight 0. Hence, for each $j\in [k]$ we get a $s\leadsto t$ path whose weight is exactly equal to $\alpha$. However, now the canonical paths will encounter exactly one vertex in each gadget: either green or orange depending on whether the gadget is symmetric or asymmetric respectively. As shown in Figure~\ref{fig:savings-new}, we can save 2 in every symmetric gadget and 1 in every asymmetric gadget. Since the number of symmetric gadgets is $k$ and the number of asymmetric gadgets is $k(k-1)$, we save a total weight of $2k+k(k-1) = (k^{2}+k)$.

Hence, the total weight of the solution is equal to $\Big(\alpha \cdot k + W(k-1) \Big) + \Big(\alpha\cdot k - (k^{2}+k)\Big) = \beta$, from Equation~\ref{eqn:beta}.
\end{proof}

\subsection{\pmb{\blaha} has a solution of weight $\pmb{\leq \beta \Rightarrow}$ \textsc{Grid Tiling}* has a solution}
\label{subsec:hard-direction}

We now prove the other direction which is more involved. Fix a solution $\mathcal{X}$ of \blaha which has cost $\leq \beta$. First we show some preliminary lemmas:

\begin{definition}
\label{defn-level}
For each $i\in [k]$, let us call the set of gadgets $\{G_{i,1}, G_{i,2}, \ldots, G_{i,k}\}$ as the gadgets of level $i$.
\end{definition}

\begin{lemma}
\label{lem:t-s-path} The $t\leadsto s$ path in $\mathcal{X}$
\begin{itemize}
\item Must use all the $k-1$ connector edges
\item Contains an $a_i \leadsto b_i$ path (for each $i\in [k]$) which does not include any connector edge
\end{itemize}
\end{lemma}
\begin{proof}
The only outgoing edge from $t$ is $(t, a_k)$ and the only incoming edge into $s$ is $(b_1, s)$. Hence, the $t\leadsto s$ is essentially a path from $a_k \leadsto b_1$. Since the edges in the gadgets are oriented downwards and rightwards, the only way to reach a gadget of level $i-1$ from a gadget of level $i$ is to go to the vertex $b_i$ and then use the path $b_{i}\rightarrow e_i \rightarrow f_i \rightarrow a_{i-1}$. That is, we must use all the $(k-1)$ connector edges which are given by $(e_i, f_i)$ for each $2\leq i\leq k$.

Now we show the second part of the lemma. First observe that the above argument also implies that $\mathcal{X}$ contains an $a_{i}\leadsto b_i$ path for each $2\leq i\leq k$. Since the only incoming edge into $s$ is $(b_1, s)$, we must also have an $a_{1}\leadsto b_1$ path in $\mathcal{X}$. Therefore, the $t\leadsto s$ path contains an $a_i \leadsto b_i$ path for each $i\in [k]$. Clearly, the $a_k \leadsto b_k$ path cannot use any connector edge due to orientation of the edges. For $1\leq i\leq k-1$ consider a $a_i \leadsto b_i$ path $P$ in $\mathcal{X}$. If it uses any connector edge, say $(e_j, f_j)$, then it follows from the orientation of the edges that $j>i$. Hence this path $P$ reaches the vertex $b_j$ which is at level $j$. Recall that the only way to climb a level above in the graph (that is, one with a smaller index) is through connector edges. Therefore, the next time that the path $P$ reaches level $i$ (which it has to in order to reach vertex $b_i$) it must do so at vertex $a_i$. Hence, the $a_i\leadsto b_i$ sub-path of $P$ which starts at the last occurrence of $a_i$ on $P$ satisfies the condition of not using any connector edge.

\end{proof}

\begin{lemma}
\label{lem:no-s-t-share-blue-with-a-b} For each $i\in [k]$, the sum of weights of blue edges incident on $a_i$ and $b_i$ on the $a_i \leadsto b_i$ path in $\mathcal{X}$ is at least $\Delta(nk+1)$.
\end{lemma}
\begin{proof}
From Lemma~\ref{lem:t-s-path}, for each $i\in [k]$ we know that $\mathcal{X}$ contains an $a_i \leadsto b_i$ path which does not include any connector edge, i.e., the edges of this $a_i \leadsto b_i$ path are contained among the gadgets of level $i$. We must use at least one blue edge incident on $a_i$ and one blue edge incident on $b_i$. Let the blue edges incident on $a_i, b_i$ be from the canonical paths $P^{\ell}_{i}, P^{\ell'}_{i}$. Since the edges in gadgets are oriented downwards and rightwards, it follows that $\ell' \geq \ell$. Hence the sum of weights of the blue edges is given by $\Delta(nk-ni+n+1-\ell) + \Delta(ni-n+\ell') = \Delta(nk+1)+ (\ell'-\ell) \geq \Delta(nk+1)$.
\end{proof}

\begin{lemma}
\label{lem:connector-edges}
At least $k$ of the $s\leadsto t$ paths in $\mathcal{X}$ cannot use any connector edge.
\end{lemma}
\begin{proof}
If less than $k$ of the $s\leadsto t$ paths in $\mathcal{X}$ do not use connector edges, then this implies that at least $k$ of the $s\leadsto t$ paths in $\mathcal{X}$ use a connector edge, since we require $(2k-1)$ paths from $s\leadsto t$. %
Since there are exactly $(k-1)$ connector edges, some connector edge is used by two different $s\leadsto t$ paths. As we have seen in Lemma~\ref{lem:t-s-path}, the $t\leadsto s$ path in $\mathcal{X}$ must use all the $(k-1)$ connector edges. Hence, the weight of $\mathcal{X}$ is $\geq W(k-1)+ W = Wk$. We show below that this is greater than $\beta$, which gives a contradiction.

\begin{align*}
\beta &= W(k-1) + 2k\Big(\Delta(nk+1) + 4(k+1)+4k(n-1)\Big) - (k^{2}+k)\\
&\leq W(k-1) + 2k\Big(\Delta(nk+1) + 4(k+1)+4k(n-1)\Big) \\
&= W(k-1) + 2k\Big(7n^{6}(nk+1) + 4(k+1)+4k(n-1)\Big) \quad \text{[Since $\Delta=7n^{6}$]}\\
&\leq W(k-1)+ 2n\Big(7n^{6}(2n^{2}) + 4(2n)+4n^{2}\Big) \quad \text{[Since $k\leq n$]}\\
&\leq W(k-1)+ 2n\Big(14n^{8} + 8n^{8}+4n^{8}\Big) \\
&= W(k-1)+ 52n^{9} \\
&< W(k-1) + 53n^{9} \\
&= W(k-1) + W \quad \text{[Since $W=53n^{9}$]}
\end{align*}
\end{proof}

We call the $s\leadsto t$ paths described in Lemma~\ref{lem:connector-edges} as \textbf{expensive} paths. Since these paths do not use a connector edge, it follows that the only outgoing edges from $s$ to be considered are to $\{c_1, c_2, \ldots, c_k\}$ and the only incoming edges into $t$ to be considered are from $\{d_1, d_2, \ldots, d_k\}$. So, we can think of the expensive paths as actually $k$ paths from $\{c_1, c_2, \ldots, c_k\}$ to $\{d_1, d_2, \ldots, d_k\}$. Since expensive paths do not use any connector edge, the existence of a $c_j\leadsto d_{\ell}$ path implies $\ell \geq j$.

\begin{definition}
\label{defn-lambda-mu}
For each $i\in [k]$, let $\lambda_i$ denote the number of $c_{i}\leadsto d_{i}$ expensive paths and $\mu_i$ denote the number of $c_{i}\leadsto \{d_{i+1}, d_{i+2}, \ldots, d_{k}\}$ expensive paths respectively in $\mathcal{X}$.
\end{definition}

From Lemma~\ref{lem:connector-edges}, it follows that
\begin{equation}
\label{eqn:sum-lambda-mu}
\sum_{i=1}^{k} (\lambda_i +\mu_i) \geq k
\end{equation}

\begin{lemma}
\label{lem:a-i-b-i} Let $c_j\leadsto d_{\ell}$ be an expensive path in $\mathcal{X}$. Then the sum of weights of the blue edges in this path is exactly $\Delta(nk+1)$ if the path is canonical, and at least $\Delta(nk+1) + \Delta$ otherwise.
\end{lemma}
\begin{proof}
Since expensive paths do not use connector edges, we have $\ell \geq j$. We consider two cases: either $\ell=j$ or $\ell>j$.

If $\ell =j$, then let the blue edges incident on $c_j, d_j$ be from the canonical paths $Q_{j}^{r}, Q_{j}^{r'}$ respectively. Since expensive paths do not use connector edges, we have $r' \geq r$. The weight of blue edges incident on $c_j$ from canonical path $Q_{j}^{r}$ is $\Delta(nk-nj+n+1-r)$ and the weight of the blue edge incident on $d_j$ from the canonical path $Q_{j}^{r'}$ is $\Delta(nj-n+r')$. Hence, the sum of weights of these edges is $\Delta(nk-nj+n+1-r)+ \Delta(nj-n+r') = \Delta(nk+1) + \Delta(r'-r)) \geq \Delta(nk+1)$. Note that if the path is canonical then $r'=r$ and the weight is exactly $\Delta(nk+1)$. Otherwise, if the path is not canonical then $r'>r$ and then the weight is $\Delta(nk+1) + \Delta(r'-r)) \geq \Delta(nk+1)+\Delta$

In the last case suppose $\ell>j$: so clearly the path is not canonical. The minimum weights of any blue edges incident on $c_j, d_{\ell}$ are $\Delta(nk-nj+1), \Delta(n\ell-n+1)$ respectively. Hence, the sum of weights of these edges is $\Delta(nk-nj+1) + \Delta(n\ell-n+1) = \Delta(nk+1) + \Delta + \Delta(n(\ell-j-1)) \geq \Delta(nk+1) + \Delta$.

\end{proof}

\begin{lemma}
\label{lem:blue} The weight of blue edges in $\mathcal{X}$ is at least $2k\cdot \Delta(nk+1)$.
\end{lemma}
\begin{proof}
From Lemma~\ref{lem:no-s-t-share-blue-with-a-b}, we know that the sum of weights of blue edges incident on $a_i$ and $b_i$ on the $a_i \leadsto b_i$ path in $\mathcal{X}$ is at least $\Delta(nk+1)$ for each $i\in [k]$. From Lemma~\ref{lem:a-i-b-i}, we know that the sum of weights of the blue edges in any expensive path is at least $\Delta(nk+1)$. Moreover, Lemma~\ref{lem:connector-edges} implies that there are at least $k$ expensive paths.
Since all these edges are clearly disjoint, it follows that the total weight of blue edges is at least $2k\cdot \Delta(nk+1)$.
\end{proof}

\begin{lemma}
\label{lem:black} The weight of black edges in $\mathcal{X}$ is at least $2k\Big(4(k+1)+4k(n-1)\Big)$, without considering the savings via orange and green edges (see Figure~\ref{fig:savings-new}).
\end{lemma}
\begin{proof}

From Lemma~\ref{lem:t-s-path}, we know that for each $i\in [k]$ there is an $a_i \leadsto b_i$ path in $\mathcal{X}$ which does not include any connector edge. Hence, the edges of this $a_i \leadsto b_i$ paths are contained in the gadgets of level $i$. Hence, we need to at least buy the set of horizontally right  black edges which take us from $a_i$ to $b_i$. These black edges have weight $4(k+1)+4k(n-1)$. Since the edges of the $a_i \leadsto b_i$ paths are contained in the gadgets of level $i$ and the sets of horizontally right black edges in gadgets of different levels are disjoint, the total weight of horizontally right black edges is at least $k\Big(4(k+1)+4k(n-1)\Big)$.

Similarly, let $c_{j}\leadsto d_{\ell}$ be an expensive path for some $\ell\geq j$. Again, we need to at least buy at least the set of vertically downward black edges which take us from $c_j$ to $d_{\ell}$. These vertically downward black edges have total weight $4(k+1)+4k(n-1)$. Even though two expensive paths may use the same vertically downward edges, they are both to be used in $s\leadsto t$ paths and hence we must pay for them each time. By Lemma~\ref{lem:connector-edges}, there are at least $k$ expensive paths and hence the total weight of the vertically downward black edges is at least $k\Big(4(k+1)+4k(n-1)\Big)$.

Combining the two observations above, it follows that the total weight of black edges (horizontally right and vertically downward) in $\mathcal{X}$ is at least $2k\Big(4(k+1)+4k(n-1)\Big)$, without considering the savings via orange and green edges (see Figure~\ref{fig:savings-new}).
\end{proof}

\begin{lemma}
\label{lem:canonical-lambda} Every expensive path in $\mathcal{X}$ is canonical, i.e., $\mu_j = 0$ for all $j\in [k]$.
\end{lemma}
\begin{proof}
Suppose an expensive path is not canonical. Lemma~\ref{lem:a-i-b-i} implies that the contribution of the blue edges of this expensive path is $\geq \Delta(nk+1)+\Delta$. By an argument exactly similar to that of Lemma~\ref{lem:blue}, it follows that the contribution of the blue edges to the weight of $\mathcal{X}$ is at least $2k\cdot \Delta(nk+1) + \Delta$.

Refer to Figure~\ref{fig:savings-new}. Note that we can use each shortcut at most $\binom{2k}{2}$ times, once for each pair of paths that will meet at the orange or green vertex (note that there are total $2k$ paths ). There are $k\cdot n$ green edges ($n$ in each of the $k$ symmetric gadgets). Since each green shortcut can save a weight of 2, we can save at most $2k\cdot n$ from the green edges. Note that in the asymmetric gadgets, there are no shortcuts along the diagonal. Hence, an asymmetric gadget can have at most $(n^{2}-n)$ orange edges. There are $(k^{2}-k)$ asymmetric gadgets and we can save a weight of 1 from each orange edge. So, we can save at most $(n^{2}-n)(k^{2}-k)$ from the orange edges. Hence, total maximum saving is at most
\begin{align*}
\binom{2k}{2}\Big(2k\cdot n + (n^{2}-n)(k^{2}-k)\Big)
&= k(2k-1)\Big(2k\cdot n + (n^{2}-n)(k^{2}-k)\Big) \\
&\leq 2n^2\cdot (2n^2 + n^4) & \text{[Since $k\leq n$]} \\
&\leq 6n^{6}
\end{align*}
We now claim that the weight of our solution already exceeds $\beta$, even if we allow this maximum possible saving. Recall that we have weight of $W(k-1)$ from the connector edges. Hence, the weight of $\mathcal{X}$ is at least
\begin{align*}
\text{weight of}\ \mathcal{X} &\geq  W(k-1) + \Big(2k\cdot \Delta(nk+1) + \Delta \Big)+ 2k\Big(4(k+1)+4k(n-1)\Big) - 6n^6 \\
&= W(k-1) + 2k\cdot \Delta(nk+1) + 2k\Big(4(k+1)+4k(n-1)\Big) + \Big(\Delta - 6n^6\Big) \\
&= W(k-1) + 2k\cdot \Delta(nk+1) + 2k\Big(4(k+1)+4k(n-1)\Big) + n^6 \quad \text{[Since $\Delta=7n^{6}$]} \\
&> W(k-1) + 2k\cdot \Delta(nk+1) + 2k\Big(4(k+1)+4k(n-1)\Big) \\
&> W(k-1) + 2k\cdot \Delta(nk+1) + 2k(4(k+1)+4k(n-1)) - (k^{2}-k) \\
&= \beta \quad \text{[From Equation~\ref{eqn:beta}]}
\end{align*}
Contradiction.
\end{proof}

\begin{lemma}
\label{lem:sum-lambdas}
$\sum_{i=1}^{k} \lambda_i = k$
\end{lemma}
\begin{proof}
From Equation~\ref{eqn:sum-lambda-mu} and Lemma~\ref{lem:canonical-lambda}, it follows that $\sum_{i=1}^{k} \lambda_i \geq k$. Suppose $\sum_{i=1}^{k} \lambda_i > k$, i.e., there are at least $k+1$ expensive paths. We follow a line of argument similar to that in proof of Lemma~\ref{lem:blue}. Note that the blue edges incident used in an expensive path are not used in the $t\leadsto s$ path in $\mathcal{X}$, and hence it follows that the total cost of the blue edges from expensive paths is at least $(k+1)\cdot \Delta(nk+1)$. Hence the total weight of the blue edges is at least $k\cdot \Delta(nk+1) + (k+1)\cdot \Delta(nk+1) = 2k\cdot \Delta(nk+1) + \Delta(nk+1) \geq 2k\cdot \Delta(nk+1) + \Delta$. Now an argument similar to that of Lemma~\ref{lem:canonical-lambda} shows that the weight of $\mathcal{X}$ exceeds $\beta$, which is a contradiction.
\end{proof}

Note the shortcuts described in Figure~\ref{fig:savings-new} again bring the $a_i\leadsto b_i$ path back to the same horizontal
canonical path.
\begin{definition}
\label{defn-almost-canonical} We call an $a_i\leadsto b_i$ path as an \emph{almost-canonical} path if it is basically a
canonical path, but can additionally take the small detour given by the green or orange edges in Figure~\ref{fig:savings-new}. An almost-canonical path must however end on the same horizontal level on which it began.
\end{definition}

\begin{lemma}
\label{lem:min-wt}
$\mathcal{X}$ contains exactly one $a_i \leadsto b_i$ path for each $i\in [k]$. Moreover, this path is almost-canonical.
\end{lemma}
\begin{proof}
Fix some $i\in [k]$. From Lemma~\ref{lem:t-s-path}, we know that $\mathcal{X}$ contains an $a_i \leadsto b_i$ path which does not include any connector edge, i.e., this path is completely contained in the gadgets of level $i$.
Suppose to the contrary that the $a_{i}\leadsto b_{i}$ path in $\mathcal{X}$ is not almost-canonical. From the orientation of the edges in the gadgets of level $i$ (rightwards and downwards), we know that there is a $a_{i}\leadsto b_i$ path in $\mathcal{X}$ that starts with the blue edge from $P_{i}^{\ell}$ and ends with a blue edge from $P_{i}^{\ell'}$ for some $\ell' > \ell$. Hence, the contribution of these blue edges is $\Delta(nk-ni+n+1-\ell) + \Delta(ni-n+\ell') = \Delta(nk+1) + \Delta (\ell'-\ell) \geq \Delta(nk+1) + \Delta$. Now, a similar argument as in Lemma~\ref{lem:canonical-lambda} can be applied to show that the weight of $\mathcal{X}$ is greater than $\beta$. Contradiction.

The above paragraph shows that each $a_i\leadsto b_i$ path in $\mathcal{X}$ is almost-canonical. Suppose there are at least two $a_i\leadsto b_i$ paths in $\mathcal{X}$. Then the blue edges incident on $a_i, b_i$ must be different (otherwise we get the same almost-canonical path). Therefore, the sum of blue edges incident on $a_i$ and $b_i$ is $\geq 2\Delta(nk+1)$. A similar argument to Lemma~\ref{lem:canonical-lambda} can be applied to show that the weight of $\mathcal{X}$ is greater than $\beta$. Contradiction.
\end{proof}

\begin{theorem}
\label{thm:dsf-redn-hard} If OPT for \blah is at most $\beta$ then the \textsc{Grid Tiling}* instance has a solution.
\end{theorem}
\begin{proof}
By Lemma~\ref{lem:sum-lambdas}, we know that $\sum_{i=1}^{k} \lambda_i = k$ and $\lambda_i \geq 0$ for each $i\in [k]$. We now claim that $\lambda_i = 1$ for each $i\in [k]$. By Lemma~\ref{lem:canonical-lambda} and Lemma~\ref{lem:min-wt}, we know that $\mathcal{X}$ contains
\begin{itemize}
\item (Property 1): Exactly one $a_i\leadsto b_i$ (almost-canonical) path for every $i\in [k]$.
\item (Property 2): Exactly $k$ canonical expensive paths.
\end{itemize}
In addition, $\mathcal{X}$ contains $(k-1)$ connector edges. For the moment let us forget the shortcuts we added in Figure~\ref{fig:savings-new}. The weight of $\mathcal{X}$, without considering the shortcuts from Figure~\ref{fig:savings-new}, is equal to $W(k-1)+ 2k\Big(\Delta(nk+1)+4(k+1)+4k(n-1)\Big) = \beta + (k^{2}+k)$. Therefore, we must have a saving of $\geq (k^{2}+k)$ from the orange and green shortcuts.

By Lemma~\ref{lem:min-wt}, we know that for each $i\in [k]$ there is exactly one $a_i\leadsto b_i$ path in $\mathcal{X}$. Moreover, this path is almost-canonical. Recall that only the horizontal edges can save some weight (see Figure~\ref{fig:savings-new}). Therefore, we can use at most $k$ green edges (one for each symmetric gadget). Each canonical expensive path can use at most $(k-1)$ orange edges; once for each of the $(k-1)$ asymmetric gadgets that it encounters along the way. Suppose we use $\delta$ green edges. Then Property 1 and Property 2 above imply that $\delta\leq k$. Then the total saving is at most $(k-1)\Big(\sum_{i=1}^{k} \lambda_i) + 2\delta = k(k-1)+ 2\delta$. Since we want the total saving to be at least $k(k-1)+2k$, this forces $\delta\geq k$. But, we already know that $\delta\leq k$, and hence $\delta=k$. This forces that $\lambda_i = 1$ for each $i\in [k]$ as follows: If any $\lambda_i = 0$, then we cannot use the green edge in the symmetric gadget $G_{i,i}$ which contradicts $\delta=k$. If any $\lambda_i \geq 2$, then some other $\lambda_j = 0$ (since $\sum_{i=1}^{k} \lambda_i = k$) and we return to previous case. Therefore, the total saving is exactly $k(k-1)+2k$.

So, we have that for each $j\in [k]$, there is a canonical $c_{j}\leadsto d_{j}$ path in $\mathcal{X}$, say $Q_{j}^{\gamma_j}$. Further, $\mathcal{X}$ also contains an $a_i \leadsto b_i$ almost-canonical path for any $i\in [k]$, say $P_{i}^{\alpha_i}$. The fact that we have a saving of at least $k(k-1)+2k$ implies we have exactly one intersection in each symmetric gadget and each non-symmetric gadget. By construction of the gadgets, it follows that
\begin{itemize}
\item $\gamma_i= \alpha_i$ for each $i\in [k]$,
\item for each $1\leq i\neq j\leq k$ we have $(\alpha_i, \alpha_j)\in S_{i,j}$.
\end{itemize}
That is, the set of values $(\alpha_i, \alpha_j)\in S_{i,j}$ for each $1\leq i,j \leq k$ form a solution for the \textsc{Grid Tiling}* instance.
\end{proof}

\subsection{Proof of Theorem~\ref{thm:main-hardness}}

Finally, we are now ready to prove Theorem~\ref{thm:main-hardness} which is restated below:
\begin{reptheorem}{thm:main-hardness}
The \blah problem is W[1]-hard parameterized by $k$. Moroever, under the ETH, the \blah problem cannot be solved in $f(k)\cdot n^{o(k)}$ time for any function $f$ where $n$ is the number of vertices in the graph.
\end{reptheorem}
\begin{proof}
Theorem~\ref{thm:dsf-redn-easy} implies the W[1]-hardness by giving a reduction which transforms the problem of $k\times k$ \textsc{Grid Tiling}* into an
instance of \blaha where we want to find $2k-1$ paths from $s\leadsto t$ and one path from $t\leadsto s$.

Chen et al.~\cite{chen-hardness} showed for any computable function $f$, the existence of an $f(k)\cdot n^{o(k)}$ time algorithm for \textsc{Clique} implies ETH fails. Composing the reduction of~\cite{daniel-grid-tiling} from \textsc{Clique} to \textsc{Grid Tiling}*, along with our reduction from \textsc{Grid Tiling}* to \blaha, we obtain under ETH there is no $f(k)\cdot n^{o(k)}$ algorithm for \blah for any computable function $f$ since the parameter blowup is linear. This shows that the $n^{O(k)}$ algorithm for \blah given in Section~\ref{sec:algorithm} is asymptotically optimal.
\end{proof}

\section{Conclusions}
\label{sec:conclusions}

In this paper, for any $k\geq 1$ we studied the \blah problem and presented an algorithm which finds an optimum solution in $n^{O(k)}$ time. Moreover, we showed our algorithms is asymptotically optimal: under the ETH, the \blah problem does not admit an $f(k)\cdot n^{o(k)}$ time algorithm for any computable function $f$. This algorithm crucially used the fact that there always exists an optimal solution for \blah that has the reverse-compatibility property. However, we showed in Section~\ref{app:no-structure} that the \blahgeneral problem need not always have an optimal solution which satisfies the general-reverse-compatibility property when $\min\{k_1, k_2\}\geq 2$. Therefore, \blah is the most general problem that one can attempt to solve with our techniques. It remains an important challenging problem to find a similar structure and generalize our method to solve the \blahgeneral problem.

~\\

\noindent \textbf{Acknowledgements}: We would like to thank DIMACS for its hospitality where a subset of the authors had fruitful discussions on this problem.



%
%
%
%
%
%
%

\bibliographystyle{splncs03}
\bibliography{docsdb}

\begin{thebibliography}{10}
\providecommand{\url}[1]{\texttt{#1}}
\providecommand{\urlprefix}{URL }

\bibitem{DBLP:conf/ipco/ChakrabartyCKK11}
Chakrabarty, D., Chekuri, C., Khanna, S., Korula, N.: Approximability of
  capacitated network design. In: Integer Programming and Combinatoral
  Optimization - 15th International Conference, {IPCO} 2011, New York, NY, USA.
  pp. 78--91 (2011)

\bibitem{DBLP:journals/jal/CharikarCCDGGL99}
Charikar, M., Chekuri, C., Cheung, T.Y., Dai, Z., Goel, A., Guha, S., Li, M.:
  {Approximation Algorithms for Directed Steiner Problems}. Journal of
  Algorithms  33(1) (1999)

\bibitem{chen-hardness}
Chen, J., Huang, X., Kanj, I.A., Xia, G.: {Strong Computational Lower Bounds
  via Parameterized Complexity}. Journal of Computer and System Sciences
  72(8),  1346--1367 (2006)

\bibitem{ipec}
Chitnis, R.H., Esfandiari, H., Hajiaghayi, M., Khandekar, R., Kortsarz, G.,
  Seddighin, S.: {A Tight Algorithm for Strongly Connected Steiner Subgraph on
  Two Terminals with Demands (Extended Abstract)}. In: Parameterized and Exact
  Computation - 9th International Symposium, {IPEC} 2014, Wroclaw, Poland. pp.
  159--171 (2014)

\bibitem{soda14}
Chitnis, R.H., Hajiaghayi, M., Marx, D.: {Tight Bounds for Planar Strongly
  Connected Steiner Subgraph with Fixed Number of Terminals (and Extensions)}.
  In: Proceedings of the Twenty-Fifth Annual {ACM-SIAM} Symposium on Discrete
  Algorithms, {SODA} 2014, Portland, Oregon, USA. pp. 1782--1801 (2014)

\bibitem{feldman-ruhl}
Feldman, J., Ruhl, M.: {The Directed Steiner Network Problem is Tractable for a
  Constant Number of Terminals}. SIAM Journal on Computing  36(2),  543--561
  (2006)

\bibitem{DBLP:conf/soda/GoemansGPSTW94}
Goemans, M.X., Goldberg, A.V., Plotkin, S.A., Shmoys, D.B., Tardos, {\'E}.,
  Williamson, D.P.: {Improved Approximation Algorithms for Network Design
  Problems}. In: Proceedings of the Fifth Annual {ACM-SIAM} Symposium on
  Discrete Algorithms. SODA 1994, Arlington, Virginia, USA. pp. 223--232 (1994)

\bibitem{guo2011hybrid}
Guo, C., Lu, G., Li, D., Wu, H., Shi, Y., Zhang, D., Zhang, Y., Lu, S.: Hybrid
  butterfly cube architecture for modular data centers (Nov~22 2011),
  \url{https://www.google.com/patents/US8065433}, {US Patent 8,065,433}

\bibitem{DBLP:conf/stoc/HalperinK03}
Halperin, E., Krauthgamer, R.: Polylogarithmic inapproximability. In:
  Proceedings of the 35th Annual {ACM} Symposium on Theory of Computing, STOC
  2003, San Diego, CA, {USA}. pp. 585--594 (2003)

\bibitem{eth}
Impagliazzo, R., Paturi, R.: {On the Complexity of $k$-SAT}. Journal of
  Computer and System Sciences  62(2),  367--375 (2001)

\bibitem{daniel-grid-tiling}
Marx, D.: {On the Optimality of Planar and Geometric Approximation Schemes}.
  In: 48th Annual {IEEE} Symposium on Foundations of Computer Science {(FOCS}
  2007), Providence, RI, USA. pp. 338--348 (2007)

\bibitem{DBLP:conf/icalp/Marx12}
Marx, D.: {A Tight Lower Bound for Planar Multiway Cut with Fixed Number of
  Terminals}. In: Automata, Languages, and Programming - 39th International
  Colloquium, {ICALP} 2012, Warwick, UK. Part {I}. pp. 677--688 (2012)

\bibitem{michal-stacs}
Marx, D., Pilipczuk, M.: {Everything you always wanted to know about the
  parameterized complexity of Subgraph Isomorphism (but were afraid to ask)}.
  In: 31st International Symposium on Theoretical Aspects of Computer Science
  {(STACS} 2014), Lyon, France. pp. 542--553 (2014)

\bibitem{daniel-voronoi}
Marx, D., Pilipczuk, M.: Optimal parameterized algorithms for planar facility
  location problems using voronoi diagrams. In: Algorithms - {ESA} 2015 - 23rd
  Annual European Symposium, Patras, Greece. pp. 865--877 (2015)

\bibitem{ramachandran2010wireless}
Ramachandran, K., Kokku, R., Mahindra, R., Rangarajan, S.: {Wireless Network
  Connectivity in Data Centers} (Jul~8 2010),
  \url{http://www.google.com/patents/US20100172292}, {US Patent App.
  12/499,906}

\bibitem{ramanathan1996multicast}
Ramanathan, S.: Multicast tree generation in networks with asymmetric links.
  IEEE/ACM Transactions on Networking (TON)  4(4),  558--568 (1996)

\bibitem{DBLP:conf/sigmetrics/TeixeiraMSV03}
Teixeira, R., Marzullo, K., Savage, S., Voelker, G.M.: {Characterizing and
  measuring path diversity of internet topologies}. In: International
  Conference on Measurements and Modeling of Computer Systems, {SIGMETRICS}
  2003, San Diego, CA, {USA}. pp. 304--305 (2003)

\bibitem{DBLP:conf/imc/TeixeiraMSV03}
Teixeira, R., Marzullo, K., Savage, S., Voelker, G.M.: {In search of path
  diversity in ISP networks}. In: 3rd {ACM} {SIGCOMM} Internet Measurement
  Conference, {IMC} 2003, Miami Beach, FL, USA. pp. 313--318 (2003)

\end{thebibliography}


    \end{document}